\documentclass[11pt,preprint,noinfoline]{imsart}
\usepackage{amsmath,amssymb}
\usepackage{amsthm}
\usepackage{graphicx}
\usepackage{subfigure}
\usepackage{url}
\usepackage{multirow}
\usepackage[font=footnotesize]{subfig}
\numberwithin{equation}{section}
\numberwithin{table}{section}
\numberwithin{figure}{section}

\usepackage{float}
\usepackage{array}
\usepackage{mathdots}
\usepackage{cellspace}
\usepackage[round]{natbib}
\usepackage[usenames,dvipsnames,svgnames,table]{xcolor}

\usepackage{a4wide}
\allowdisplaybreaks

\newtheorem{theorem2}{Theorem}[section]
\newtheorem{proposition}{Proposition}[section]
\newtheorem{lemma}{Lemma}[section]
\newtheorem{corollary}{Corollary}

\newtheorem{remark}{Remark}[section]

\newcommand\vv{\tilde{\nu}}
\newcommand\vvv{\hat{\nu}}
\newcommand\VV{\tilde{V}}
\newcommand\VVV{\hat{V}}
\newcommand\hab{\hat{\beta}}

\newcommand\SSS{\hat{S}}

\newcommand\cvp{\stackrel{p}{\longrightarrow}}

\newcommand\cvd{\stackrel{d}{\longrightarrow}}

\newcommand\inn[2]{\langle #1, #2 \rangle}

\begin{document}

\newcommand{\red}{\color{DarkRed}}
\newcommand\cov{\mathop{\text{cov}}}
\newcommand\var{\mathop{\text{var}}}
\newcommand\diag{\mathop{\text{diag}}}

\begin{frontmatter}
  \title{Extension of the Lagrange multiplier test for error
    cross-section independence to  large panels  with non normal    errors}
  \runtitle{Extension of the Lagrange multiplier test to general large panel models}
  
  \thankstext{m1}{Z. Li's research is partially supported by a
    National Natural Science Foundation of China grant (No. 11901492),
    and The CUHK-SZ Presidential Fund (PF01001160).}
  \thankstext{m3}{Jianfeng Yao's  research is partially supported by a Hong Kong SAR RGC
    Grant (GRF 17308920).}

  \begin{aug} 
    \author{\fnms{Zhaoyuan} \snm{Li}\thanksref{m1}\ead[label=e1]{lizhaoyuan@cuhk.edu.cn}}
    \and
    \author{\fnms{Jianfeng} \snm{Yao}\thanksref{m3}\ead[label=e3]{jeffyao@hku.hk}}

    \runauthor{Z. Li and J. Yao}
    
    \address{School of Data Science \\  The Chinese University of Hong Kong, Shenzhen \\ Shenzhen, China \\   \printead{e1}}
    \address{Department of Statistics and Actuarial Science\\
      The    University of Hong Kong\\
      Hong Kong SAR,~~~   China\\
      \printead{e3}}
  \end{aug} 

  \begin{abstract}
    This paper reexamines the seminal  Lagrange multiplier test  for
    cross-section independence in a large panel  model
    where both the number of cross-sectional units $n$ and
    the number of time series observations $T$
    can be large.
    The first contribution of the paper is an enlargement of the 
    test with two extensions: firstly the new asymptotic normality is
    derived
    in a  simultaneous  limiting scheme where the two dimensions
    $(n,T)$ tend to infinity
    with  comparable magnitudes; second, the result is valid for
    general error distribution (not necessarily normal).
    The second contribution of the  paper is a new test statistic
    based on the sum of the  fourth
    powers of cross-section correlations from OLS residuals, instead
    of their squares  used in the Lagrange multiplier statistic. This
    new test is  generally more powerful, and the improvement
    is particularly visible against alternatives with weak or sparse cross-section
    dependence.
    Both  simulation study and real data analysis are proposed to
    demonstrate the advantages of the enlarged Lagrange multiplier test and the
    power enhanced test in comparison with  the existing procedures.
  \end{abstract}
  

\end{frontmatter}

\section{Introduction}\label{sec:intro}

Consider the fixed effects panel data model
\begin{eqnarray}\label{model}
  y_{it} = \alpha + x_{it}^\prime \beta +\mu_i +\nu_{it}, \quad \textrm{for} \ i=1,\ldots,n;\ t=1,\ldots,T,
\end{eqnarray}
where $i$ indexes the cross-sectional (individual) units, and $t$ the
time series observations. The dependent variable is $y_{it}$ and
$x_{it}$ denotes the exogenous regressors of dimension $k\times 1$,
$\beta$ is the corresponding $k\times 1$ vector of parameters, $\mu_i$
denotes the time-invariant individual effect which could be correlated with the regressors $x_{it}$. Throughout the paper the number of covariates $k$ is fixed while
the two dimensions $n$ and $T$ may grow to infinity.
A more general model is 
the following heterogeneous panel  model
\begin{eqnarray}\label{model_gene}
  y_{it} = x_{it}^\prime \beta_i +\nu_{it}, \ \textrm{for} \ i=1,\ldots,n;\ t=1,\ldots,T,
\end{eqnarray}
where the slope parameters, $\beta_i$, are  allowed to vary across $i$.

Our focus is to  test  the following cross-section independence
hypothesis, that is
\begin{equation}\label{hypothesis}
  H_0: \quad  \text{the model errors}
  ~~ \{\nu_{it}\} ~~ \text{are independent across the units~~}
  1\le i\le n.
\end{equation}

Testing such  cross-section independence is important because it is a
preliminary step for  many existing inference procedures for the panel model.
These procedures become biased or even inconsistent when
cross-sectional units are correlated, see for example
\citep{chudik2013large} for the case of commonly used panel unit root tests.

Several procedures exist in the literature for this independence test. A popular procedure  is  the Lagrange multiplier (LM) test proposed
by \cite{breusch1980lagrange}
based on the sum of squared
pair-wise sample correlation coefficients of the OLS residuals. Recent efforts have concentrated in large  panel models
where the panel size $n$ is large compared to the sample size $T$. It
has been observed that the LM test is particularly biased in
such large panels.
\cite{pesaran2008bias} found an approximation for the mean and
variance of the LM statistic and established its asymptotic normality
for their modified $LM_{adj}$ test under the following  {\em sequential limit
  scheme}
\[  \text{SEQ-L:}\qquad
T \to \infty \text{~~first,  followed by}~~     n\to \infty.
\]   
\cite{pesaran2004general} also proposed a  cross-sectional dependence (CD)
test using sum of (non-squared) sample correlations of the OLS residuals. \cite{pesaran2015testing} showed that the null of the CD test is
rather weak cross-sectional dependence as defined in
\cite{chudik2011}. Again the asymptotic normality for the test
statistics is derived under the sequential limit scheme.

When the panel size $n$ and the sample size $T$ are large while having
comparable magnitude, it has been argued in recent high-dimensional statistic literature that
more reliable asymptotic results can be found by considering a
{\em simultaneous limit scheme} \citep{YZB15-book}
\[
\text{SIM-L:}\qquad T \to \infty, n=n(T) \quad  \text{such that }
\lim_{T\to\infty} \frac{n}{T} = c>0.
\] 
Technically, the derivation of asymptotic normality for LM-type
statistics under the SIM-L scheme is challenging. Actually,   
\cite{baltagi2012lagrange}  succeeded  for the  fixed effects
homogeneous panel data model \eqref{model} by precisely identifying a
mean shift in the asymptotic distribution of the LM statistic due to
the SIM-L scheme.
This bias-corrected test, $LM_{bc}$,  however needs to assume the normality of
the errors $(\nu_{it})$ as its development relies  on a previous
result established in \cite{schott2005testing}
on sample correlations of  normal-distributed  errors.
As for the heterogeneous model \eqref{model_gene},
a recent work
\cite{jiangworking2020} establishes the asymptotic normality
of the LM test under the SIM-L scheme while also requiring
normally distributed errors $(\nu_{it})$   as well
as normally distributed regressors $(x_{it})$. This test,  
$LM_{RMT}$, is extensively compared to  the CD and  $LM_{adj}$ tests in various
dimension and panel size combinations via simulation. 
The $LM_{RMT}$ has been shown to be comparable to $LM_{adj}$ in terms
of size and power, with a slight preference for $LM_{RMT}$ when the
sample size $T$ is not very large (or when the panel size $n$ is
relatively large).
The CD test is universally correctly sized but generally lacks power
under local alternatives.  

An important aim of this paper to extend the LM test to large panel
models under the SIM-L scheme {\em and} without the normality assumption on the errors $(\nu_{it})$.
Because the normality assumption is fundamental to the analytic 
methods used in both the $LM_{bc}$ and $LM_{RMT}$ tests, we
achieve our goal by employing complete different tools and derive
a new asymptotic normal distribution for the LM statistic in
Theorem~\ref{th:trace_residual}.
A corresponding test $LM_e$ is thus derived for cross-section
independence for the panel model~\eqref{model}.
We also establish a surprising fact that this asymptotic distribution
for the $LM_e$ coincide with those of both the $LM_{bc}$ and
$LM_{RMT}$ tests although  the technical derivations of the three
tests are all different.  This universality of the asymptotic
distribution for the LM statistic under the SIM-L scheme ensures a
welcomed robustness of the LM test, thus extending its
application scope to various large panel models.

A second contribution from the paper is to go beyond the  LM
statistic by considering fourth powers of residual correlations
(instead of their squares in the LM statistic).  The idea here is that
fourth powers weight more heavily large correlations than
 smaller ones.  This magnification effect enhances the power of the
 test, particularly under an alternative where the cross-section dependence is weak or
sparse. Following the same setting and tools as for the LM statistic, we establish
the asymptotic normality for this  power-enhanced
test (PET).

Next we designed several simulation experiments following commonly
used settings in the large panel model literature. These experiments
confirm the excellent empirical performance of the new $LM_e$ test at
a level comparable to the tests $LM_{RMT}$  and $LM_{adj}$.  When it
comes to compare the powers of the tests, the power-enhanced
test PET indeed dominates all the existing tests discussed above. As
expected,  the 
advantage of the PET  test is particularly significant under  weak
or sparse dependence alternatives.

We have also conducted  a real data analysis to assess the properties of
these tests in real-life large panels. This findings basically confirm
the comparison results found through simulation experiments.


The rest of the paper is organised as follows. Section~\ref{sec:OLS}
presents OLS regression residuals in the panel model \eqref{model} and
the existing tests discussed above.  Sections~\ref{sec:newtest1} and \ref{sec:newtest2} 
introduce, respectively,   the extended LM test   $LM_e$ and the
power-enhanced test  PET,  and  establish their
asymptotic normality under the null and the SIM-L scheme without
assuming normality of the errors.    
Section~\ref{sec:simul} presents a detailed simulation study
for comparison of finite-sample performance of the various tests.
A real data analysis is carried out in Section~\ref{sec:realdata}.
Some discussions are offered in the last conclusion section.
All technical proofs are grouped to the Appendix.

\section{OLS regression residuals and existing tests}
\label{sec:OLS}
 
Consider the population correlation matrix
$R=[\diag(\Sigma)]^{-1/2} \Sigma  [\diag(\Sigma)]^{-1/2}$ of the
error vectors $\nu_t$.
Clearly, under the independence hypothesis~\eqref{hypothesis},
one has  $R=I$.
It is thus  natural to deign test statistics based on estimates for
these   error correlations.   Natural estimates for these correlations
are obtained using the residuals  from some  fitted panel regression
model. Precisely, 
consider   the centralised variables
\[ \tilde{y}_{it}=y_{it}-\frac{1}{T}\sum_{t=1}^T y_{it}, \quad
\tilde{x}_{it}=x_{it}-\frac{1}{T}\sum_{t=1}^Tx_{it}, \quad
\text{and}\quad 
\tilde{\nu}_{it}=\nu_{it}-\frac{1}{T}\sum_{t=1}^T \nu_{it}.
\]
The model~\eqref{model} takes a simpler form with these centralised
variables,
\[
\tilde y_{it} = \tilde  x_{it}^\prime \beta  +\tilde\nu_{it}, \quad
\textrm{for} \ 1\le i\le n;\ 1\le t\le T.
\]
The  OLS estimator of the regression  parameter $\beta$ in model \eqref{model} is 
\begin{eqnarray*}
  \hat{\beta}= \left(\sum_{t=1}^T \sum_{i=1}^n \tilde{x}_{it}\tilde{x}_{it}^\prime \right)^{-1} \left(\sum_{t=1}^T \sum_{i=1}^n \tilde{x}_{it}\tilde{y}_{it} \right).
\end{eqnarray*}
Introduce also the design matrices
\[
X_i=\begin{pmatrix}\tilde{x}_{i1},...,\tilde{x}_{it},...,\tilde{x}_{iT}\end{pmatrix}_{k\times  T},\quad 
X=\begin{pmatrix} X_1,X_2,...,X_n\end{pmatrix}_{k\times nT},
\]
and  the staked observation vectors,
\[
Y_i=\begin{pmatrix} \tilde{y}_{i1}, \cdots,\tilde{y}_{iT} \end{pmatrix}^{\prime},\quad
Y=\begin{pmatrix} Y^{\prime}_{1} , \cdots,Y^{\prime}_{n} \end{pmatrix}^{\prime}.
\]
Then $\hab=\left(XX^{\prime}\right)^{-1}XY$.  It is well known that
under fairly general assumptions on the
panel model and  the independence hypothesis $H_0$, $\hat{\beta}$ is a consistent estimator of $\beta$ no matter whether $n$ is fixed or tends to infinity jointly with $T$ \citep{baltagi2011testing,baltagi2017asymptotic}. The regression residuals $\hat{\nu}_{it}$ are 
\begin{eqnarray}\label{ols_residual}
  \hat{\nu}_{it} = \tilde{y}_{it} - \tilde{x}_{it}^\prime \hat{\beta} = \tilde{\nu}_{it}-\tilde{x}_{it}^\prime \left(\hat{\beta}-\beta \right).
\end{eqnarray}
Let  $\hat{\nu}_t = (\hat{\nu}_{1t},\ldots,\hat{\nu}_{nt})^\prime$ for
$t=1,\ldots, T$.
The {\em sample residual covariance matrix}  $\hat{S}_T$ and
the {\em sample residual correlation matrix}  $\hat{R}_T$ are
respectively,
\begin{equation}\label{hatD-T}
  \hat{S}_T = \frac{1}{T}\sum_{t=1}^T \hat{\nu}_t \hat{\nu}_t^\prime,
  \quad \text{and}\quad
  \hat{R}_T=\hat{D}_T^{-1/2}\hat{S}_T\hat{D}_T^{-1/2}, 
  \quad\text{with}\quad
  \hat{D}_T=\diag(\hat{S}_T).
\end{equation}
 The entries of  $\hat{R}_T $ are denoted as
  $\hat{R}_T=\{\hat\rho_{rs}, ~1\le r,s\le n\}$.

Introduce also  the 
{\em sample error covariance and correlation matrix} defined
analogously but with the centralised errors $\{\tilde{\nu}_t\}$, that is
\begin{equation}\label{D-T}
S_T=\frac{1}{T}\sum_{t=1}^T \tilde{\nu}_t \tilde{\nu}_t^\prime,
\quad \text{and}\quad
R_T=D_T^{-1/2}S_TD_T^{-1/2},
\quad\text{with}\quad D_T=\diag(S_T).
\end{equation}

\subsection{The Lagrange Multiplier test ($LM$)}
\label{LMtest}

This very  first test for cross-section  independence is  
proposed by \cite{breusch1980lagrange}  and uses the
statistic 
\begin{equation}\label{LM}
  LM = \frac{T}{2}\left\{tr (\hat{R}_T^2)-n\right\} =
  \frac{T}{2} \sum_{r\ne s} \hat\rho_{rs}^2.
\end{equation}
Note that   $\sum_{r\ne s} \hat\rho_{rs}^2=||\hat{R}_T-I||_F^2 $ is
the squared Frobenius distance between ${\hat R}_T$ and $I$,
the $LM$ statistic is indeed a scaled estimator for the population distance 
$||{R}-I||_F^2$ which is zero under the null hypothesis.

Consider  the following assumptions:
\begin{itemize}
	\item[(A1)] For each $i$, the disturbances, $\nu_{it}$, are serially independent with mean 0 and variance, $0<\sigma_i^2<\infty$;
	\item[(A2)] Under the null hypothesis defined by $H_0$ in (\ref{hypothesis}): $\nu_{it}=\sigma_i\varepsilon_{it},$ with $ \varepsilon_{it}\stackrel{i.i.d.}{\sim} N(0, 1)$ for all $i$ and $t$.
\end{itemize}

Under a {\em large sample limit scheme} where $n$ is fixed while
$T\to\infty$ and Assumptions (A1) and (A2), 
\cite{breusch1980lagrange} established that under the null,
the  $LM$ statistic  has an asymptotic $\chi^2_{n(n-1)/2}$ distribution.
It has been well discussed in the literature that this LM test  
based on its large sample limiting 
chi-squared distribution  suffers from severe size
distortions for large panels where $n$ is large compared to time series size $T$.

\subsection{The cross-section dependence test ($CD$)}

To accommodate large panels, 
\cite{pesaran2004general} proposed the following  CD statistic
\begin{eqnarray}\label{CD}
  CD  = \sqrt{\frac{T}{2n(n-1)}} \sum_{r\neq s}\hat{\rho}_{rs}. 
\end{eqnarray}
Consider  the following additional assumptions:
\begin{itemize}
	\item[(A3)]  The disturbances are $\nu_{it}=\sigma_i \varepsilon_{it}$, with $\varepsilon_{it}\stackrel{i.i.d.}{\sim} (0, 1)$ for all $i$ and $t$, the disturbances, $\varepsilon_{it}$, are symmetrically distributed around 0;
	\item[(A4)] The regressors, $x_{it}$, are strictly exogenous such that $E(\nu_{it}|X_i)=0$ for all $i$ and $t$, and $X_i^\prime X_i$ is a positive definite matrix;
	\item[(A5)] The OLS residuals, $\hat{\nu}_{it}$, defined by (\ref{ols_residual}), are not all zero.
\end{itemize}

Under the Assumptions
(A3)-(A4)-(A5), \cite{pesaran2004general} established that  the CD statistic is asymptotically standard normal
under the sequential limit scheme.
It is widely reported that the CD test enjoys a very accurate size in
general while suffering from certain loss of power when the panel-wise
correlations have varying signs.  This loss of power can be understood by the
fact that the CD statistic is averaging the sample residual correlations
$\{\hat{\rho}_{rs}\}$; so if they carry varying signs across  panel
units, this averaging may lead to certain cancellation of 
correlations of opposite signs,
thus neutralising the statistic and its power.
This fact is confirmed in one setting of the simulation experiments where
the  cross-sectional  dependence comes from   a factor model with mean zero factor loadings,
see Table~\ref{table2}.

\subsection{The bias-adjusted Lagrange Multiplier test ($LM_{adj}$)}
\label{ssec:LM_adj}

In order to adapt the LM test to large panels,  
\cite{pesaran2008bias} proposed the following   bias-adjusted version of the LM
statistic:
\begin{eqnarray}\label{LM_adj}
  LM_{adj} = \sqrt{\frac{1}{2n(n-1)}}\sum_{r\neq s}\frac{(T-k)\hat{\rho}_{rs}^2-\mu_{Trs}}{\sigma_{Trs}}, 
\end{eqnarray}
where 
\begin{eqnarray*}
  \mu_{Trs} &=& \frac{1}{T-k}tr(M_rM_s),\\
  \sigma_{Trs}^2 &=& [tr(M_rM_s)]^2 a_{1T}+ 2tr[(M_rM_s)^2]a_{2T},\\
  M_r&=&I_T - X_r^\prime(X_rX_r^\prime)^{-1}X_r,
  \qquad  X_r = (x_{r1}, \ldots, x_{rT}), \\ 
  a_{1T}&=&a_{2T}-\frac{1}{(T-k)^2}, \quad a_{2T}=\frac{3}{(T-k+2)^2}. 
\end{eqnarray*}
Under the assumptions (A1)-(A2)-(A4) (which imply the null hypothesis),
the authors proved that under the SEQ-L scheme,
$LM_{adj}$ is asymptotically standard normal. This bias-adjusted test
  indeed perform much better in large panels than the original LM
  test.
  The only known issue on this test is that because of the employed
  sequential limiting scheme, the test may be over-conservative in
  ``micro-panels'' where $T$ is quite limited compared to panel size
  $n$.

\subsection{A bias-corrected Lagrange Multiplier test ($LM_{bc}$)}
\label{sec:LM_bc}

\cite{pesaran2004general} also suggested the following scaled version
of the LM test for large panels:
\[
LM_P=CD_{lm}=\sqrt{\frac{1}{4n(n-1)}} \sum_{r\neq s}\left(T
\hat{\rho}_{rs}^{2}-1\right).
\]
Under the Assumptions (A1)-(A2) and the SEQ-L  scheme,
$LM_P$ is shown to have a standard normal limiting distribution.
 An issue with the test statistic as mentioned in
\cite{pesaran2004general} is that when $T$ is not large, the
scale adjustment made in $LM_P$ may not be accurate, thus
exhibiting  substantial size distortions.
Consequently, 
\cite{baltagi2012lagrange} considered the SEQ-L  scheme
where  a bias-corrected version was
found for the fixed effects model \eqref{model}.
Precisely, the authors established that 
under assumptions (A1)-(A2)-(A4) and the SEQ-L  scheme, 
\begin{eqnarray}\label{LM_bc}
  LM_{bc}:=LM_p-\frac{n}{2(T-1)}\stackrel{D}{\longrightarrow} N(0,1).
\end{eqnarray}
It is remarkable that the asymptotic bias in $LM_P$ is exactly
identified as $n/2(T-1)$, and the
authors demonstrated that this bias is caused  by the fact that the sample
correlations are calculated using OLS residuals,  instead of 
the (unobserved) model
errors. Note  that this bias  disappears 
in the traditional large sample scheme where the ratio $n/T\to 0$,
thus emphasising the particular large panel effect.  
The derivation of   \eqref{LM_bc}  uses the asymptotic results on 
sample error correlation matrix in \cite{schott2005testing}, thus
requiring the normality assumption of the errors, see Assumption (A2).

Simulation experiments in \cite{baltagi2012lagrange} show that the
bias-corrected $LM_{bc}$  test has an excellent finite sample performance in
comparison to the CD test and the $LM_{adj}$  test. It is particularly
recommendable for micro-panels where $T$ is relatively small.

\subsection{A Gaussian  Lagrange-Multiplier test for large panels ($LM_{RMT}$)}

Based on asymptotic results from random matrix theory literature,
\cite{jiangworking2020} introduced
another limiting distribution for the LM statistic
under the SIM-L  scheme and  for the heterogeneous panel data
\eqref{model_gene}.
Precisely, they considered 
the statistic,
\begin{equation}
  LM_{RMT}=\frac{tr(\hat{R}_T^2)-\mu_{rmt}}{\sigma_{rmt}},
\end{equation}
where, setting $c_T=n/T$ and  $\kappa = \frac{3T(T-k+2)}{(T+2)(T-k)}$, 
\begin{eqnarray*}
	\mu_{rmt} &=& n(1+c_T)+ c_T^2-c_T,\\
	\sigma_{rmt}^2 &=& 4c_T(1+2c_T)(c_T+2)-4(\kappa-1)c_T(1+c_T)^2+(\kappa-3)c_T(c_T-4)^2(c_T+1)^2.
\end{eqnarray*}

Consider the following  assumptions.
\begin{itemize}
	\item[(A6)] Within each unit $i$, the regression vectors $\{x_{it},1\leq t\leq T \}$ are i.i.d. $k$-variate normal $N(\mathbf{0}, \mathbf{\Lambda}_i)$ (centred with covariance matrix $\mathbf{\Lambda}_i$).
	\item[(A7)]  The covariates $\{x_{it}\}$ and the errors    $\{\nu_{it}\}$ are independent.
\end{itemize}
Under the assumptions (A1)-(A2) and (A6)-(A7), the null and  the SIM-L scheme,
\cite{jiangworking2020} established 
\begin{equation}
  LM_{RMT}   \stackrel{D}{\longrightarrow} N(0,1).
  \label{LMRMT_normal}
\end{equation}
The simulation experiments in this reference showed  that $LM_{RMT}$
is generally comparable to $LM_{adj}$ in terms
of size and power, with however a slight preference for $LM_{RMT}$
when the sample size $T$ is relatively small.
This result   requires normality for 
for both the model errors and
the  regressors, see Assumption (A2) and (A6).

\section{An extended  Lagrange multiplier test for large panels
  without normality assumption ($LM_e$)}\label{sec:newtest1}

As a main result of the present paper, 
we extend the test $LM_{RMT}$ to large panels where the errors are not
necessarily normal distributed. 
Because in large panels, most of the previous results on the Lagrange multiplier
rely on the normality assumption, the derivation of the asymptotic
distribution of the LM statistic here requires new tools. This is
achieved via 
recent results on sample correlation matrix from random matrix theory literature.

To introduce
the conditions on the centralised design matrices $X_i$, we need to
consider $E_k$, the subspace of $R^k$ which is orthogonal to the
constant vector ${\mathbf 1}_k$ (with all coordinates equal to 1),
that is,
\[   E_k =\{ u=(u_1,\ldots,u_k)\in R^k:~~ u_1+ \ldots+u_k=0.\}
\]

We will use the following assumptions on the panel model.

\begin{enumerate}
\item[(B1)] 
  The panel-wise error vectors $\nu_1,\ldots, \nu_T$ are i.i.d. with
  mean zero and uniformly bounder sixth moments, that is
  \[
  \sup_{i,t} E|\nu_{it}|^6\le C_2, \quad
  \text{for some positive constant } C_2.
  \]
\item[(B2)] 
  (i) The regressors $\{x_{it} \}$ are independent of the idiosyncratic
  disturbances $\{\nu_{it} \}$.

  (ii) There are positive  constants $a_1$ and $a_2$   such that  for all
  $i$, $T$ and non zero $u\in E_k$,
  \[ a_1 \|u\|^2 \le   \frac1T u' X_iX_i'u  \le  a_2 \|u\|^2 .
  \]

  (iii)
  For any non zero  vector $u\in E_k$ and $i$,
  \[
  \frac{\max\limits_{1\le t\le T}  \inn{u}{\tilde x_{it}}^2}{ u'X_iX_i'u} \to0,\quad   T\to\infty.
  \]
\end{enumerate}

The Assumption (B1) is a  standard moment condition on the errors
$\{\nu_{it} \}$ which are  serially uncorrelated
over time with a constant cross-section  covariance matrix $\Sigma=\cov(\nu_t)$.
The Assumption (B2) ensures the regularity of the design matrices with
the centralised regression variables $\{\tilde x_{it}\}$.
The independence between regressors and errors required in
Assumption~(B2) is slightly stronger than the strict
exogeneousity in Assumption (A4).

\begin{theorem2}\label{th:trace_residual}
Suppose Assumptions  (B1) and (B2)  hold for the panel data model
(\ref{model}). Then under the SIM-L scheme and the null hypothesis, 
\begin{eqnarray}\label{LM_c}
  LM_e:=\frac{tr(\hat{R}_{T}^2)-\mu_{LM_e}}{\sigma_{LM_e}}
  \stackrel{D}{\longrightarrow} N(0,1), 
\end{eqnarray}
with
\begin{eqnarray}
  && \mu_{LM_e} = n(1+c_T)+c_T^2-c_T, \label{mu_LM}\\
  && \sigma_{LM_e}^2=4c_T^2.  \label{sigma_LM}
\end{eqnarray}
\end{theorem2}

This result is established in two steps. Using recent results from
random matrix theory,  we first find that
\begin{equation}
  \frac{tr({R}_{T})-\mu_{LM_e}}{\sigma_{LM_e}}
  \stackrel{D}{\longrightarrow} N(0,1).
  \label{trace_error}
\end{equation}
Next, we show that 
\begin{equation}
  tr (\hat R_T^2)-  tr(R_T^2) = o_p(1).
  \label{trace_residual}
\end{equation}
The first step result \eqref{trace_error} is justified in
Appendix~\ref{app:trace1}. The justification of the second step result
\eqref{trace_residual} requires more calculations; they are detailed
in Appendix~\ref{sec:trace4}, see Proposition~\ref{trace2_prop}.

%
%

We now compare the new test $LM_e$ to two existing tests also
developed in the SIM-L scheme.  With respect to the bias-corrected
test $LM_{bc}$, note that 
the $LM_p$ statistic (Section~\ref{sec:LM_bc})  can be rewritten as
\[  LM_p = 
\frac{ tr(\hat{R}_T^2)-\mu_{LM_p} }{ \sigma_{LM_p}}
\]
with
\[ 
\mu_{LM_p}= n(1+c_T)-c_T+o(1), \quad
\text{and}\quad 
\sigma_{LM_p}^2=4c_T^2=\sigma_{LM_e}^2.
\]
We have 
\[
\frac{tr(\hat{R}_T^2)-\mu_{LM_e}}{2c_T}=\frac{tr(\hat{R}_T^2)-\mu_{LM_p}-c_T^2}{2c_T}
= LM_P -\frac{c_T}{2}.
\]
Thus by  Theorem~\ref{th:trace_residual}, 
\[LM_P -\frac{c_T}{2}  \stackrel{D}{\longrightarrow} N(0,1), \]
under the SIM-L scheme.  This coincide with the asymptotic normality
of $LM_{bc}$ given in  (\ref{LM_bc}). 
Therefore, our $LM_e$ test can be seen as an extension of the
$LM_{bc}$ test to panels with non normal-distributed errors.

Next we compare the $LM_{e}$ test  to the  result \eqref{LMRMT_normal} for
the Gaussian $LM_{RMT}$   test.   We  find that $\mu_{rmt}=\mu_{LM_e}$, and
$\sigma_{rmt}^2\sim  \sigma_{LM_e}^2\sim 4c^2$ for large $T$ (and
$n$) since
as $\kappa\to 3$ when $k$ is fixed and $T\to \infty$.
This  means that  
under the SIM-L scheme,
the asymptotic  distribution of $tr(\hat{R}_{T}^2)$
derived in Theorem~\ref{th:trace_residual} for the extended LM test
$LM_e$ 
is also valid for the heterogeneous model (\ref{model_gene})
assuming normality of the errors (and the regression variables).

In conclusion of these comparisons, 
the asymptotic  distribution of $tr(\hat{R}_{T}^2)$ 
derived in Theorem~\ref{th:trace_residual} has a welcomed
universality, being valid for the three tests $LM_{e}$, $LM_{adj}$ and
$LM_{RMT}$ which cover quite different large panel models. Such
distributional robustness enlarges the application scope of the LM
statistic to various large panel models .

\section{A power enhanced test for large panels}\label{sec:newtest2}

Anticipating the simulation results shown in
Section~\ref{sec:simul}, the various large-panel versions of the LM
test, namely $LM_{adj}$, $LM_{bc}$, $LM_{RMT}$ as well  as the new
test $LM_e$ may suffer from the problem of low power against large
panels where the units are weakly dependent. Such weak dependence
arises for example when the cross-sectional correlation matrix $R$
is sparse. Recent literature in high-dimensional statistics
indicates that an efficient way for detection of sparse correlations
is to weight those relatively significant sample correlations
$\hat\rho_{rs}$ more heavily 
than those small sample correlations.  An extreme method in this
regard is for example to take the overall maximum
$\max\limits_{r\ne s}|\hat\rho_{rs}|$ as a test statistic.
It is however unclear in the current SIM-L large panel setting how to
derive a limiting distribution for such maximum type statistic. Here
we propose a manageable compromise by considering the sum of 
fourth powers of the sample residuals correlations, that is, to consider
\[ 
tr(\hat{R}_T^4)=\sum\limits_{r\ne s} {\hat\rho}^4_{rs},
\]
as the new test statistic. The rational here is that 
compared to the sum of the squares $\{ {\hat\rho}^4_{rs}\}$  used
in the LM  statistic,
$tr(\hat{R}^4)$ is weighting more heavily
larger sample correlations than  smaller
ones.  This will enhance the power of the test when either very few sample
correlations are significantly non zero, or they are many but with
relatively small amplitudes.
Such situations arise under an alternative with sparse cross-section
dependence,
or with globally weak cross-dependence.

We derive the asymptotic normality of $tr(\hat{R}_T^4)$ in the
following theorem.

\begin{theorem2}\label{PET_residual}
  Suppose Assumptions~(B1) and (B2)  hold for the panel data
  model (\ref{model}).
  Then under the SIM-L  scheme and under the null $H_0$ in  (\ref{hypothesis}),
  \begin{eqnarray}\label{PET}
	\frac{tr(\hat{R}_T^4)-\mu_{PET}}{\sigma_{PET}} \stackrel{D}{\longrightarrow} N(0, 1).
  \end{eqnarray}
  with
  \begin{eqnarray}
    \mu_{PET} &=& n\left(1+\frac{6n}{T-1}+\frac{6n^2}{(T-1)^2}+\frac{n^3}{(T-1)^3} \right) -6c_T(1+c_T)^2-2c_T^2, \label{mu_PET}\\ [1mm]
    \sigma_{PET}^2 &=& 8c_T^4+96c_T^3(1+c_T)^2+16c_T^2(3c_T^2+8c_T+3)^2 \label{sigma_PET}.
  \end{eqnarray}
\end{theorem2}

The limiting normality of $tr(\hat{R}^4)$ in (\ref{PET}) under the
null allows us to perform a level-$\alpha$ test for  the null
hypothesis $H_0$.
This test is hereafter referred as the {\em power enhanced  test
  (PET)} for cross-section independence in large panels.

The proof of Theorem~\ref{PET_residual} follows the strategy for that
of Theorem~\ref{th:trace_residual} and also proceeds in two steps.
In the first step, using recent results from
random matrix theory,  we  find that
\begin{equation}
  \frac{tr({R^4}_{T})-\mu_{PET}}{\sigma_{PET}}
  \stackrel{D}{\longrightarrow} N(0,1).
  \label{PET_error}
\end{equation}
Next, we show that 
\begin{equation}
  tr (\hat R_T^4)-  tr(R_T^4) = o_p(1).
  \label{trace4_diff}
\end{equation}
The first step result \eqref{PET_error} is justified in
Appendix~\ref{sec:PET}. The justification of the second step result
\eqref{trace4_diff} also requires more calculations; they are detailed
in Appendix~\ref{sec:trace4}, see Proposition~\ref{PET_prop}.


%

%

%
\section{Simulation studies}
\label{sec:simul}

We conduct a  simulation study to investigate the finite
sample performance of the proposed  tests  $LM_e$ and $PET$. Comparisons are made with
the bias-adjusted LM test ($LM_{adj}$) in (\ref{LM_adj}) and the
cross-section dependence test (CD) in (\ref{CD}).
The original Lagrange Multiplier test (LM) in (\ref{LM}) is
  excluded due to its well-known non applicability to large panels.
The Gaussian high-dimensional Lagrange Multiplier test $LM_{RMT}$ is
also excluded in comparison  as it is equivalent to our proposed
general Lagrange multiplier test ($LM_e$), see comments after
Theorem~\ref{th:trace_residual}.

\subsection{Empirical sizes of the tests}
We consider the following data generating process proposed in  \cite{pesaran2008bias}:
\begin{eqnarray}\label{data_generate}
y_{it}=a +\sum_{l=2}^k x_{lit}\beta_{l}+\mu_i+\nu_{it},\quad  1\leq i\leq n; ~1\leq t\leq T, 
\end{eqnarray}
where $a$ and $\beta_l$ are set arbitrarily to 1 and $l$, respectively, $\mu_i\stackrel{i.i.d.}{\sim}N(1,1)$ and here $k$ is the number of regressors including the intercept $a$. The regressors are generated as
\begin{eqnarray}\label{regressor}
x_{lit}=0.6 x_{li,t-1}+\sigma_{li}u_{lit}, \quad i=1,2,\ldots,n; ~t=-49, \ldots, 0,\ldots, T; ~l=2,\ldots,k,
\end{eqnarray}
with $x_{li,-50}=0$, $\sigma_{li}^2=\tau_{li}^2/(1-0.6^2)$, $\tau_{li}^2 \stackrel{i.i.d.}{\sim} \chi^2(6)/6$ and $u_{lit}\stackrel{i.i.d.}{\sim} N(0,1)$. The first 50 observations in (\ref{regressor}) are disregarded. 
Under the null, the errors $\nu_{it}$ are assumed to be i.i.d. across
individuals and over time, that take the form 
\[\nu_{it}=\sigma_i \varepsilon_{it}, \quad 1\leq i\leq n; ~1\leq t\leq T, \]
where $\sigma_i^2\stackrel{i.i.d.}{\sim}\chi^2(2)/2$ and the
$\varepsilon_{it}$'s are {also i.i.d.} and  generated from different
populations: (i) normal, $N(0,1)$, (ii) Student-t, $t_7/\sqrt{7/5}$,
and (iii) chi-square, $(\chi^2(5)-5)/\sqrt{10}$. The normalisation's in
(ii) and (iii) are  such that these  variable have  mean zero and unit
variance. 

We explore the performance of different  tests using various  combination
of $(n,T)$
with $n\in\{50, 100, 200\}$ and $T\in\{50, 100\}$. Empirical sizes and
powers of the tests are evaluated from  2,000 independent replications. The nominal test level is 5\%.

\begin{table}[ht]
	\centering
	\caption{Empirical sizes of tests (in \%).}
	\label{table1}
	\begin{tabular}{|c|cr|ccc|ccc| }
		\hline 
	& & $k$ & \multicolumn{3}{c}{2} & \multicolumn{3}{c|}{4}\\ \cline{2-9}
	& T & n & 50 & 100 & 200 & 50 & 100 & 200\\ \hline 
	\multirow{10}{*}{Normal}&  \multirow{5}{*}{50} & $LM_e$ & 5.00 & 5.05 & 4.70 & 4.80 & 6.05 & 5.55 \\ 
	 & & PET & 5.25 & 5.65 & 5.10 & 4.70 & 5.55 & 4.70 \\
	&  & $LM_{adj}$ & 5.20 & 5.15 & 4.55 & 4.20 & 4.55 & 2.80 \\ 
	& & CD &  5.45 & 4.95 & 5.10 & 4.95 & 5.15 & 5.30 \\ \cline{2-9}
	& \multirow{5}{*}{100} & $LM_e$ & 5.45 & 5.00 & 5.30 & 4.65 & 5.25 & 5.40 \\
	& & PET & 4.80 & 4.70 & 5.40 & 4.25 & 5.70 & 5.10 \\ 
	& & $LM_{adj}$ & 5.75 & 5.05 & 5.35 & 4.70 & 5.05 & 4.30 \\ 
	& & CD & 4.45 & 5.50 & 5.40 & 5.20 & 4.85 & 5.15 \\ \hline
	\multirow{10}{*}{Student-t}& \multirow{5}{*}{50} &$LM_e$ & 5.20 & 4.90 & 4.70 & 5.15 & 5.55 & 5.35 \\
	& & PET & 5.45 & 6.20 & 5.05 & 4.75 & 5.90 & 5.20 \\ 
	& & $LM_{adj}$ & 5.55 & 5.05 & 4.20 & 4.80 & 4.10 & 2.65 \\ 
	& & CD & 5.55 & 5.35 & 4.45 & 5.10 & 5.05 & 5.55  \\ \cline{2-9}
	& \multirow{5}{*}{100} & $LM_e$& 5.00 & 4.45 & 5.50 & 4.80 & 4.90 & 5.60  \\ 
	& & PET & 4.40 & 4.25 & 5.50 & 4.95 & 5.00 & 5.35  \\ 
	& & $LM_{adj}$ & 5.30 & 4.65 & 5.50 & 4.90 & 4.75 & 5.00 \\ 
	& & CD & 4.55 & 5.30 & 5.30 & 4.80 & 5.25 & 4.60\\ \hline 
	\multirow{10}{*}{Chi-square}& \multirow{5}{*}{50} & $LM_e$ & 4.85 & 6.45 & 5.65 & 5.75 & 5.35 & 5.10 \\
	& & PET &4.90 & 6.55 & 5.25 & 5.70 & 5.65 & 4.95 \\
	& & $LM_{adj}$ & 5.45 & 6.45 & 5.30 & 5.40 & 4.15 & 2.35 \\
	& & CD & 5.20 & 4.85 & 4.90 & 5.05 & 5.00 & 4.85 \\ \cline{2-9}
   &	\multirow{5}{*}{100} & $LM_e$ & 4.70 & 5.25 & 6.35 & 6.05 & 5.15 & 6.15 \\
   	& & PET & 4.85 & 4.90 & 5.90 & 5.05 & 5.05 & 5.80 \\
	& & $LM_{adj}$ & 5.15 & 5.40 & 6.35 & 6.40 & 4.80 & 5.25 \\
	& & CD & 4.70 & 4.75 & 5.95 & 6.10 & 4.90 & 4.90 \\
	\hline
	\end{tabular}
\end{table}

Table~\ref{table1} presents the empirical sizes of $LM_e$, PET,
$LM_{adj}$ and CD tests (values close to 5\% are better) under three
different error distributions.  The proposed PET test and the CD test
have a similar  performance  a little better than $LM_e$, which is in
trun slightly better than $LM_{adj}$. However, the $LM_{adj}$
test is noticeably computationally more demanding than the other
tests. It also  has  a large downside size distortion under the cases with $T=50$, $n=200$ and $k=4$.

\subsection{Empirical powers of the tests}

To evaluate  the power of the tests considered, the disturbances are
generated according to the following single-factor model,
\[\nu_{it} =  \lambda_{i}f_{t} +\epsilon_{it},\quad 1\leq i\leq n; ~1\leq t\leq T, \]
where $\lambda_{i}$ is the factor loading of individual $i$ for the
common factor $f_{t}$ in period $t$, with
$f_t\stackrel{i.i.d.}{\sim}N(0,1)$.
The factor loadings are constructed under   three different scenarios:
\begin{itemize}
\item Dense case:
   here the strength of cross-sectional correlation is measured
    by a positive parameter $h>0$. Given $h$,
  $\lambda_i\stackrel{i.i.d.}{\sim} U[-b,b]$ for $i=1,\ldots,n$, where
  $b=\sqrt{3h/n}$ (the average of the squared length
  $\lambda_1^2+\cdots+\lambda_n^2$ is thus $h$);
\item Sparse case: $\lambda_i\stackrel{i.i.d.}{\sim}U(0.5,1.5)$ for
  $i=1,2,\ldots, [n^{0.3}]$ and $\lambda_i=0$ for $i=[n^{0.3}]+1,
  \ldots,n$, where $[n^{0.3}]$ is the integer part of $n^{0.3}$.
  We have $[n^{0.3}]= 3, 3, 4$ for $n=50, 100, 200$, respectively.
	\item Less-sparse case:
      $\lambda_i\stackrel{i.i.d.}{\sim}U(0.5,1.5)$ for $i=1,2,\ldots,
      [n^{0.5}]$ and $\lambda_i=0$ for $i=[n^{0.5}]+1, \ldots,n$.
      We      have  $[n^{0.5}]= 7, 10, 14$ for $n=50, 100, 200$, respectively.
\end{itemize}
In the dense case, all cross-sectional units are correlated.
 The correlation between units $i$ and $i'$ is $\lambda_i\lambda_{i'}$
  and the overall strength of correlation is controlled by $h$.
We study the empirical powers of the tests when $h$ varies  while
remaining bounded, that is, $h/n\to 0$.  This thus corresponds to the
setting of weak factor alternative used in
\cite{baltagi2017asymptotic}. In the  sparse case, only a few, about
$n^{0.3}$, cross-sectional units are correlated while other units are
uncorrelated.
When this  number of correlated cross-sectional units increases to
$n^{0.5}$,  we call it a less-sparse case.

Table~\ref{table2}  shows the empirical powers of the tests in the 
dense case. The dimensions $n$, $T$ and $k$ are set to be (50, 100), 100 and 2,
respectively. The cross-sectional correlation strength  $h$ varies
from 1 to 7.  As expected, the empirical powers of all tests increase
with the strength  $h$.  The PET test largely outperforms the
others. Compared to $LM_e$ and $LM_{adj}$, PET indeed boosts the power
by up to 36\% for cases with small value of $h$. $LM_{adj}$
performs a little better than $LM_e$. The CD test has very low powers, confirming the fact that its implicit null  is rather weak
cross-sectional dependence \citep{pesaran2004general}.
Plots at the bottom of Table~\ref{table2}  illustrate the evolution of
these powers in function of the varying strength $h$
and for the cases with chi-square distributed errors.

Tables~\ref{table3} and \ref{table4} show the empirical powers for the
sparse and the less-sparse cases, respectively. As expected,  all the tests have
higher power  in the less-sparse case  than in the sparse cases.
The proposed PET test again performs best. It  boosts the power by up
to 16\% under these sparse cases as  compared to $LM_e$ and
$LM_{adj}$. The proposed $LM_e$ test performs better than $LM_{adj}$.
The latter  has   low powers in the  cases with $T=50$ and $k=4$,
which is consistent with the conservative sizes observed in
Table~\ref{table1}. The CD test has no powers for the  sparse case,
but has some powers for the less-sparse case; it has  however an
overall poor performance compared to the other tests.

\begin{table}[ht]
\centering
\caption{Empirical powers of tests under dense case with  $n=50, 100$, $T=100$, $k=2$ and varying $h$.}
\label{table2}
\begin{tabular}{|c|r|ccccccc|}
	\hline 
	n & $h$ & 1 & 2 & 3 & 4 & 5 & 6 & 7 \\
	\hline 
	\multirow{15}{*}{50} & & \multicolumn{7}{c|}{Normal}\\ \cline{2-9}
 &	$LM_e$ & 0.3610 & 0.8445 & 0.9930 & 1 & 1 & 1 & 1  \\
&	PET & 0.4800 & 0.9375 & 0.9990 & 1 & 1 & 1 & 1\\
&	$LM_{adj}$ & 0.3685 & 0.8475 & 0.9935 & 1 & 1 & 1 & 1 \\
&	CD & 0.0580 & 0.0505 & 0.0480 & 0.0485 & 0.0465 & 0.0450 & 0.0455 \\
	\cline{2-9}
&	& \multicolumn{7}{c|}{Student-t}\\ \cline{2-9}
&	$LM_e$ & 0.3060 & 0.8245 & 0.9910 & 1 & 1 & 1 & 1  \\
&    PET & 0.4095 & 0.9195 & 0.9990 & 1 & 1 & 1 & 1\\
&	$LM_{adj}$ & 0.3155 & 0.8315 & 0.9915 & 1 & 1 & 1 & 1 \\
&	CD & 0.0490 & 0.0475 & 0.0500 & 0.0595 & 0.0725 & 0.0850 & 0.0960  \\
	\cline{2-9}
&	& \multicolumn{7}{c|}{Chi-square}\\ \cline{2-9}
&	$LM_e$ & 0.3415 & 0.8490 & 0.9920 & 1 & 1 & 1 & 1  \\
&	PET & 0.4510 & 0.9400 & 0.9990 & 1 & 1 & 1 & 1 \\
&	$LM_{adj}$ & 0.3550 & 0.8545 & 0.9920 & 1 & 1 & 1 & 1 \\
&	CD & 0.0475 & 0.0450 & 0.0430 & 0.0525 & 0.0620 & 0.0760 & 0.0930  \\
	\hline 
	\multirow{15}{*}{100} & & \multicolumn{7}{c|}{Normal}\\ \cline{2-9}
	&	$LM_e$ & 0.1255 & 0.5035 & 0.9005 & 0.9860 & 0.9990 & 1 & 1  \\
	&	PET & 0.1500 & 0.6765 & 0.9600 & 0.9990 & 1 & 1 & 1 \\
	&	$LM_{adj}$ & 0.1285 & 0.5080 & 0.9040 & 0.9860 & 0.9990 & 1 & 1 \\
	&	CD & 0.0520 & 0.0440 & 0.0460 & 0.0475 & 0.0510 & 0.0555 & 0.0570 \\
	\cline{2-9}
	&	& \multicolumn{7}{c|}{Student-t}\\ \cline{2-9}
	&	$LM_e$ & 0.1140 & 0.5285 & 0.9205 & 0.9910 & 0.9995 & 1 & 1  \\
	&    PET & 0.1380 & 0.7205 & 0.9795 & 0.9990 & 1 & 1 & 1 \\
	&	$LM_{adj}$ & 0.1185 & 0.5350 & 0.9220 & 0.9910 & 0.9995 & 1 & 1 \\
	&	CD & 0.0495 & 0.0445 & 0.0460 & 0.0480 & 0.0480 & 0.0490 & 0.0510 \\
	\cline{2-9}
	&	& \multicolumn{7}{c|}{Chi-square}\\ \cline{2-9}
	&	$LM_e$ & 0.1140 & 0.5730 & 0.9355 & 0.9970 & 1 & 1 & 1  \\
	&	PET & 0.1490 & 0.7525 & 0.9850 & 0.9995 & 1 & 1 & 1 \\
	&	$LM_{adj}$ & 0.1195 & 0.5835 & 0.9380 & 0.9975 & 1 & 1 & 1 \\
	&	CD & 0.0580 & 0.0380 & 0.0440 & 0.0445 & 0.0480 & 0.0485 & 0.0515  \\
	\hline 
\end{tabular}
\vskip2mm Plots for the Chi-square cases \\
\includegraphics[width=0.43\textwidth]{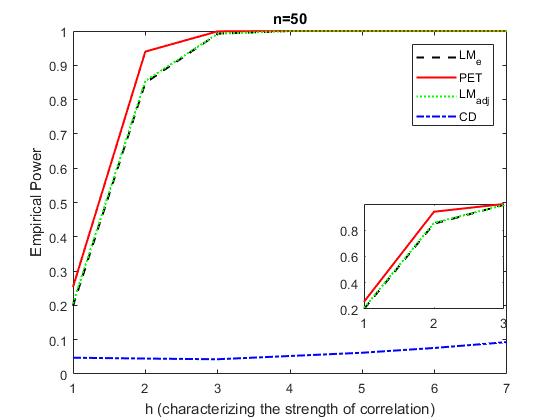}
\includegraphics[width=0.43\textwidth]{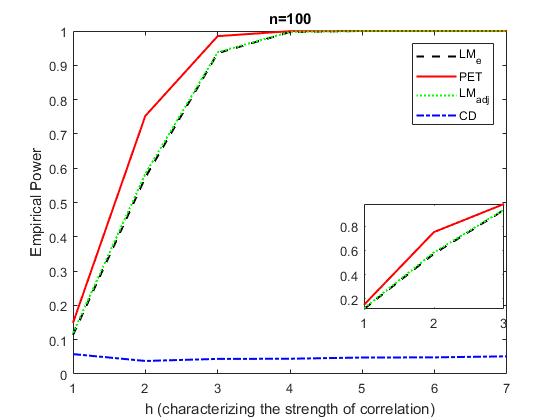}
\end{table}

\begin{table}[ht]
	\centering
	\caption{Empirical powers of tests under sparse case.}
	\label{table3}
	\begin{tabular}{|cc|ccc|ccc|ccc|}
		\hline 
		& & \multicolumn{3}{c|}{Normal} & \multicolumn{3}{c|}{Student-t} & \multicolumn{3}{c|}{Chi-square}\\ \hline
		T & n & 50 & 100 & 200 & 50 & 100 & 200 & 50 & 100 & 200 \\ \hline
		\multicolumn{11}{c}{$k=2$}\\ \hline 
		\multirow{4}{*}{50} & $LM_e$ & 0.1990 & 0.1195 & 0.1005 & 0.2010 & 0.1005 & 0.0955 & 0.1415 & 0.1315 & 0.1125  \\
		& PET & 0.2195 & 0.1240 & 0.1035 & 0.2175 & 0.1065 & 0.0920 & 0.1370 & 0.1355 & 0.1160  \\
		& $LM_{adj}$ & 0.2085 & 0.1230 & 0.0955 & 0.2130 & 0.1010 & 0.0915 & 0.1480 & 0.1345 & 0.1075 \\
		& CD & 0.0830 & 0.0620 & 0.0660 & 0.0790 & 0.0605 & 0.0640 & 0.0690 & 0.0620 & 0.0580 \\ \hline
		\multirow{4}{*}{100} & $LM_e$ & 0.5735 & 0.2045 & 0.1805 & 0.4320 & 0.1970 & 0.2095 & 0.1695 & 0.1205 & 0.1165  \\
		& PET & 0.6625 & 0.2295 & 0.1960 & 0.5050 & 0.2225 & 0.2360 & 0.1810 & 0.1315 & 0.1255  \\
		& $LM_{adj}$ & 0.5815 & 0.2095 & 0.1815 & 0.4455 & 0.2025 & 0.2110 & 0.1790 & 0.1255 & 0.1175 \\
		& CD & 0.0865 & 0.0650 & 0.0660 & 0.0785 & 0.0620 & 0.0660 & 0.0655 & 0.0525 & 0.0510 \\ \hline 
		\multicolumn{11}{c}{$k=4$}\\ \hline 
		\multirow{4}{*}{50} & $LM_e$ & 0.1405 & 0.0755 & 0.0840 & 0.1345 & 0.1370 & 0.1095 & 0.2300 & 0.1160 & 0.1095  \\
		& PET & 0.1475 & 0.0835 & 0.0970 & 0.1445 & 0.1475 & 0.1125 & 0.2425 & 0.1145 & 0.1125  \\
		& $LM_{adj}$ & 0.1320 & 0.0570 & 0.0495 & 0.1270 & 0.1040 & 0.0570 & 0.2240 & 0.0915 & 0.0580 \\
		& CD & 0.0660 & 0.0660 & 0.0580 & 0.0690 & 0.0685 & 0.0580 & 0.0755 & 0.0580 & 0.0580 \\ \hline
		\multirow{4}{*}{100} & $LM_e$ & 0.2890 & 0.1210 & 0.1345  & 0.4655 & 0.1460 & 0.1030 & 0.1785 & 0.2955 & 0.1545  \\
		& PET & 0.3350 & 0.1225 & 0.1395 & 0.5410 & 0.1640 & 0.1070 & 0.1920 & 0.3540 & 0.1835  \\
		& $LM_{adj}$ & 0.2910 & 0.1175 & 0.1155 & 0.4705 & 0.1400 & 0.0865& 0.1805 & 0.2870 & 0.1395 \\
		& CD & 0.0830 & 0.0705 & 0.0590 & 0.0915 & 0.0740 & 0.0560 & 0.0760 & 0.0715 & 0.0695 \\ \hline 
	\end{tabular}
\vskip2mm Plots for the Chi-square cases \\
\includegraphics[width=0.47\textwidth]{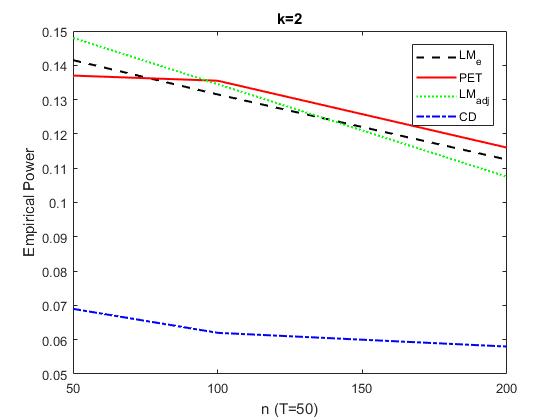}
\includegraphics[width=0.47\textwidth]{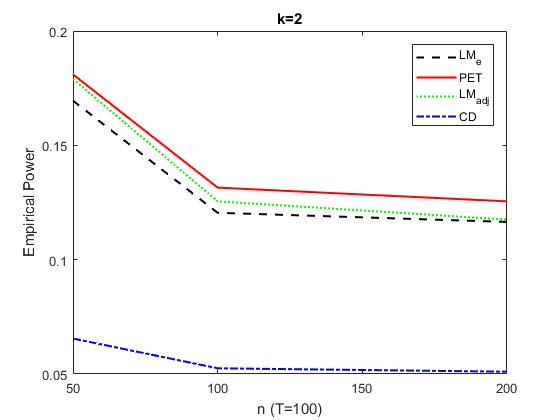}\\ 
\includegraphics[width=0.47\textwidth]{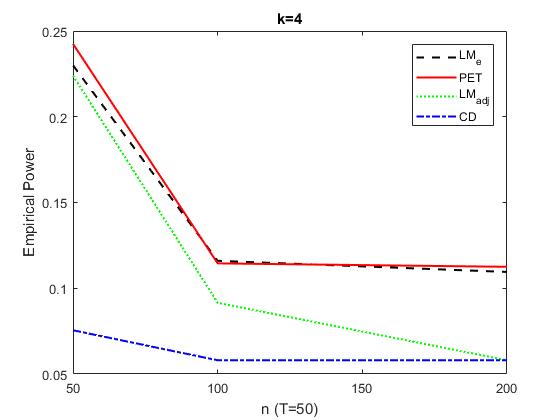}
\includegraphics[width=0.47\textwidth]{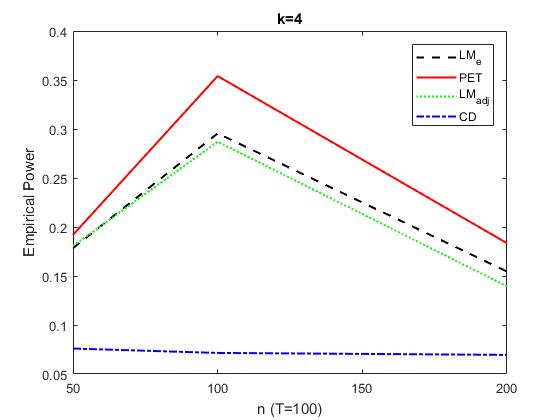}
\end{table}

\begin{table}[ht]
	\centering
	\caption{Empirical powers of tests under less-sparse case.}
	\label{table4}
	\begin{tabular}{|cc|ccc|ccc|ccc|}
		\hline 
		& & \multicolumn{3}{c|}{Normal} & \multicolumn{3}{c|}{Student-t} & \multicolumn{3}{c|}{Chi-square}\\ \hline
		T & n & 50 & 100 & 200 & 50 & 100 & 200 & 50 & 100 & 200 \\ \hline
		\multicolumn{11}{c}{$k=2$}\\ \hline 
		\multirow{4}{*}{50} & $LM_e$ & 0.9780 & 0.9870 & 0.9910 & 0.9805 & 0.9610 & 0.9750 & 0.9850 & 0.9985 & 0.9135  \\
		& PET & 0.9965 & 0.9980 & 0.9990 & 0.9960 & 0.9920 & 0.9935 & 0.9965 & 1 & 0.9745  \\
		& $LM_{adj}$ & 0.9800 & 0.9870 & 0.9910 & 0.9820 & 0.9610 & 0.9735 & 0.9865 & 0.9985 & 0.9100 \\
		& CD & 0.5135 & 0.5330 & 0.5535 & 0.5110 & 0.4815 & 0.4990 & 0.5190 & 0.6380 & 0.4120 \\ \hline
		\multirow{4}{*}{100} & $LM_e$ &  1 & 1 & 1 &  1 & 1 & 1 & 0.9965 & 1 & 1  \\
		& PET & 1 & 1 & 1 & 1 & 1 & 1 & 0.9995 & 1 & 1  \\
		& $LM_{adj}$ & 1 & 1 & 1 & 1 & 1 & 1 & 0.9975 & 1 & 1 \\
		& CD & 0.7750 & 0.7755 & 0.8240 & 0.8630 & 0.8665 & 0.7355 & 0.5685 & 0.7940 & 0.8120 \\ \hline 
		\multicolumn{11}{c}{$k=4$}\\ \hline 
		\multirow{4}{*}{50} & $LM_e$ &  0.9265 & 0.9530 & 0.9845 & 0.9665 & 0.9810 & 0.9860 & 0.9940 & 0.9855 & 0.9945  \\
		& PET & 0.9735 & 0.9885 & 0.9960 & 0.9890 & 0.9955 & 0.9980 & 0.9995 & 0.9980 & 0.9995  \\
		& $LM_{adj}$ & 0.9240 & 0.9445 & 0.9670 & 0.9640 & 0.9790 & 0.9675 & 0.9935 & 0.9810 & 0.9870 \\
		& CD & 0.4225 & 0.4725 & 0.5220 & 0.4810 & 0.5165 & 0.5420  & 0.5490 & 0.5020 & 0.5585 \\ \hline
		\multirow{4}{*}{100} & $LM_e$ & 1 & 1 & 1  & 1& 0.9995  & 1 & 1 & 1 & 1   \\
		& PET & 1 & 1 & 1 & 1 & 1 & 1 & 1 & 1 & 1  \\
		& $LM_{adj}$ &1 & 1 & 1 & 1 &0.9995 &  1 &  1 & 1 & 1 \\
		& CD & 0.6315 & 0.7305 & 0.7725 & 0.7045 & 0.6950 & 0.8245 & 0.6530 & 0.8515 & 0.7460 \\ \hline 
	\end{tabular}
\vskip2mm Plots for the Chi-square cases \\
\includegraphics[width=0.49\textwidth]{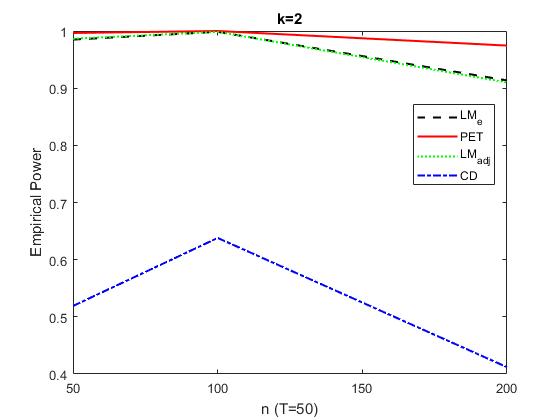}
\includegraphics[width=0.49\textwidth]{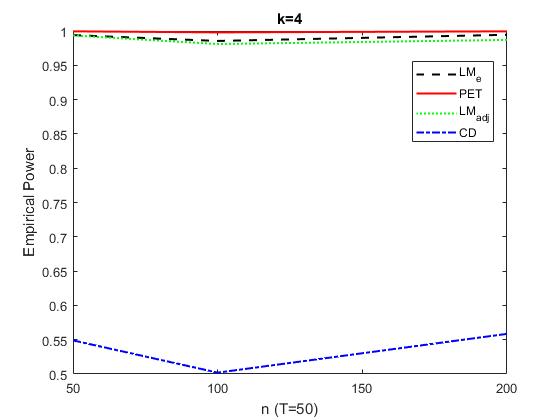}
\end{table}

\subsection{Additional simulation experiments}

The  two proposed tests $LM_e$ and $PET$ are established  for the
fixed effects panel data model \eqref{model}.  Additional  simulations
are conducted to show the finite sample performance of the proposed
tests for the heterogeneous panel data model \eqref{model_gene} although this model is not covered by the developed theory.
We consider the following 
data  generating process of \cite{pesaran2008bias}:
\[y_{it}=a +\sum_{l=2}^k x_{lit}\beta_{li}+\mu_i+\nu_{it}, \quad 1\leq i\leq n; ~1\leq t\leq T, \] 
with $\beta_{li}\stackrel{i.i.d.}{\sim} N(1,0.04)$. All other settings
are kept the same as for the fixed effect model. It is striking and satisfactory to observe that for the
  heterogeneous panel model,  
  conclusions from the simulation experiments are in general very
  similar to those reported in the previous section for the
  fixed-effect panel model. The following tables report some empirical
  results for the case of $k=2$ (the other results with $k=4$ are very
  similar and thus omitted).
Table~\ref{table5} presents the empirical sizes of all tests under
three different error distributions.
The empirical sizes of the considered  tests are all close to the
nominal level. The proposed  tests $LM_e$ and PET perform slightly
better than $LM_{adj}$. Tables~\ref{table6}, \ref{table7} and
\ref{table8} show the empirical powers of the tests for the  dense
case, the sparse case and the less-sparse case, respectively. As
expected, $LM_e$, PET and $LM_{adj}$ all show higher powers  when the
correlation matrix of errors becomes denser. Again  the proposed PET
has consistently the highest power.
In conclusion, the proposed two tests seem also valid  for the
heterogeneous panel data model (\ref{model_gene}) according to these experiments.

\begin{table}[ht]
	\centering
	\caption{Empirical sizes of tests (in \%) under heterogeneous panel data model \eqref{model_gene} with $k=2$.}
		\label{table5}
	\begin{tabular}{|cc|ccc|ccc|ccc|}
		\hline 
		& & \multicolumn{3}{c|}{Normal} & \multicolumn{3}{c|}{Student-t} & \multicolumn{3}{c|}{Chi-square}\\ \hline
		T & n & 50 & 100 & 200 & 50 & 100 & 200 & 50 & 100 & 200 \\ \hline
		\multirow{5}{*}{50} & $LM_e$ & 5.05 & 5.10 & 5.40 & 4.85 & 5.40 & 5.20 & 5.50 & 5.80 & 6.35 \\
		& PET & 4.45 & 4.95 & 5.35 & 5.30 & 5.40 & 5.35 & 5.50 & 5.50 & 5.50 \\
		& $LM_{adj}$ & 5.45 & 5.25 & 5.05 & 5.35 & 5.50 & 5.10 & 5.80 & 6.05 & 5.90 \\
		& CD & 5.25&5.00 &4.30 & 5.10 & 4.85 & 4.55 & 5.40 & 4.05 & 5.25 \\ \hline 
		\multirow{5}{*}{100} & $LM_e$ & 4.65 & 4.50 & 5.45 & 5.30 & 4.20 & 5.95 & 5.15 & 5.15 & 5.30 \\
		& PET & 4.45 & 4.85 & 5.55 & 5.45 & 4.90 & 6.25 & 5.70 & 5.90 & 5.00 \\
		& $LM_{adj}$ & 4.80 & 4.75 & 5.55 & 5.60 & 4.40 & 5.95 & 5.40 & 5.25 & 5.35 \\
		& CD & 4.25 & 5.20 & 4.85 & 4.25 & 5.10 & 4.60 & 4.45 & 4.85 & 4.75 \\
		\hline
	\end{tabular}
\end{table}

\begin{table}[ht]
	\centering
	\caption{Empirical powers of tests under dense case with $n=50, 100$, $T=100$, $k=2$ and varying $h$ for heterogeneous panel data model \eqref{model_gene}.}
	\label{table6}
	\def\arraystretch{0.7}
	\begin{tabular}{|c|r|ccccccc|}
		\hline 
		n & $h$ & 1 & 2 & 3 & 4 & 5 & 6 & 7 \\
		\hline 
		\multirow{15}{*}{50} & & \multicolumn{7}{c|}{Normal}\\ \cline{2-9}
		&	$LM_e$ & 0.2075 & 0.8490 & 0.9910 & 1 & 1 & 1 & 1  \\
		&	PET & 0.2615 & 0.9345 & 0.9995 & 1 & 1 & 1 & 1\\
		&	$LM_{adj}$ & 0.2160 & 0.8535 & 0.9910 & 1 & 1 & 1 & 1 \\
		&	CD & 0.0565 & 0.0505 & 0.0505 & 0.0490 & 0.0485 & 0.0505 & 0.0490 \\
		\cline{2-9}
		&	& \multicolumn{7}{c|}{Student-t}\\ \cline{2-9}
		&	$LM_e$ & 0.3725 & 0.8530 & 0.9915 & 1 & 1 & 1 & 1  \\
		&    PET & 0.4900 & 0.9385 & 0.9975 & 1 & 1 & 1 & 1\\
		&	$LM_{adj}$ & 0.3840 & 0.8575 & 0.9915 & 1 & 1 & 1 & 1 \\
		&	CD & 0.0645 & 0.0505 & 0.0495 & 0.0525 & 0.0530 & 0.0560 & 0.0650  \\
		\cline{2-9}
		&	& \multicolumn{7}{c|}{Chi-square}\\ \cline{2-9}
		&	$LM_e$ & 0.2140 & 0.8590 & 0.9945 & 0.9990 & 1 & 1 & 1  \\
		&	PET & 0.2650 & 0.9470 & 0.9985 & 1 & 1 & 1 & 1 \\
		&	$LM_{adj}$ & 0.2190 & 0.8625 & 0.9945 & 0.9990 & 1 & 1 & 1 \\
		&	CD & 0.0460 & 0.0560 & 0.0695 & 0.0785 & 0.0890 & 0.1010 & 0.1190  \\
		\hline 
		\multirow{15}{*}{100} & & \multicolumn{7}{c|}{Normal}\\ \cline{2-9}
		&	$LM_e$ & 0.1350 & 0.4795 & 0.8755 & 0.9865 & 0.9995 & 1 & 1  \\
		&	PET & 0.1615 & 0.6550 & 0.9665 & 0.9990 & 1 & 1 & 1 \\
		&	$LM_{adj}$ & 0.1390 & 0.4845 & 0.8785 & 0.9865 & 0.9995 & 1 & 1 \\
		&	CD & 0.0450 & 0.0435 & 0.0425& 0.0460 & 0.0470 & 0.0475 & 0.0515 \\
		\cline{2-9}
		&	& \multicolumn{7}{c|}{Student-t}\\ \cline{2-9}
		&	$LM_e$ & 0.1275 & 0.5210 & 0.9040 & 0.9900 & 0.9995 & 1 & 1  \\
		&    PET & 0.1435 & 0.6865 & 0.9705 & 0.9995 & 1 & 1 & 1 \\
		&	$LM_{adj}$ & 0.1295 & 0.5275 & 0.9065 & 0.9900 & 0.9995 & 1 & 1 \\
		&	CD & 0.0505 & 0.0420 & 0.0420 & 0.0430 & 0.0450 & 0.0475 & 0.0495 \\
		\cline{2-9}
		&	& \multicolumn{7}{c|}{Chi-square}\\ \cline{2-9}
		&	$LM_e$ & 0.1360 & 0.5660 & 0.9225 & 0.9950 & 1 & 1 & 1  \\
		&	PET & 0.1530 & 0.7310 & 0.9810 & 1 & 1 & 1 & 1 \\
		&	$LM_{adj}$ & 0.1400 & 0.5760 & 0.9245 & 0.9950 & 1 & 1 & 1 \\
		&	CD & 0.0810 & 0.0470 & 0.0495 & 0.0505 & 0.0515 & 0.0545 & 0.0545  \\
		\hline 
	\end{tabular}
\end{table}

\begin{table}[ht]
	\centering
	\caption{Empirical powers of tests under sparse case for heterogeneous panel data model \eqref{model_gene} with $k=2$.}
  \label{table7}
	\begin{tabular}{|cc|ccc|ccc|ccc|}
		\hline 
		& & \multicolumn{3}{c|}{Normal} & \multicolumn{3}{c|}{Student-t} & \multicolumn{3}{c|}{Chi-square}\\ \hline
		T & n & 50 & 100 & 200 & 50 & 100 & 200 & 50 & 100 & 200 \\ \hline
		\multirow{4}{*}{50} & $LM_e$ & 0.1405 & 0.1045 & 0.1020 & 0.1240 & 0.0715 & 0.1040 & 0.1870 & 0.1220 & 0.1290  \\
		& PET & 0.1445 & 0.1080 & 0.1035 & 0.1240 & 0.0700 & 0.1045 & 0.1925 & 0.1340 & 0.1320  \\
		& $LM_{adj}$ & 0.1480 & 0.1055 & 0.0975 & 0.1320 & 0.0720 & 0.0990 & 0.1945 & 0.1225 & 0.1200 \\
		& CD & 0.0720 & 0.0665 & 0.0645 & 0.0780 & 0.0660 & 0.0605 & 0.0735 & 0.0655 & 0.0570 \\ \hline
		\multirow{4}{*}{100} & $LM_e$ & 0.2985 & 0.1895 & 0.1840 & 0.1805 & 0.1675 & 0.1695 & 0.3555 & 0.1360 & 0.1185  \\
		& PET & 0.3205 & 0.2050 & 0.2035 & 0.1890 & 0.1855 & 0.1885 & 0.3875 & 0.1380 & 0.1250  \\
		& $LM_{adj}$ & 0.3080 & 0.1945 & 0.1850 & 0.1885 & 0.1700 & 0.1720 & 0.3625 & 0.1375 & 0.1190 \\
		& CD & 0.0800 & 0.0600 & 0.0590 & 0.0755 & 0.0545 & 0.0690 & 0.0730 & 0.0600 & 0.0545 \\ \hline 
	\end{tabular}
\end{table}

\begin{table}[ht]
	\centering
	\caption{Empirical powers of tests under less-sparse case for heterogeneous panel data model \eqref{model_gene} with $k=2$.}
		\label{table8}
	\def\arraystretch{0.7}
	\begin{tabular}{|cc|ccc|ccc|ccc|}
		\hline 
		& & \multicolumn{3}{c|}{Normal} & \multicolumn{3}{c|}{Student-t} & \multicolumn{3}{c|}{Chi-square}\\ \hline
		T & n & 50 & 100 & 200 & 50 & 100 & 200 & 50 & 100 & 200 \\ \hline
		\multirow{4}{*}{50} & $LM_e$ & 0.9030 & 0.9795 & 0.9925 & 0.9980 & 0.9620 & 0.9380 & 0.9995 & 0.9980 & 0.9830  \\
		& PET & 0.9630 & 0.9965 & 0.9985 & 1 & 0.9930 & 0.9850 & 1 & 1 & 0.9960  \\
		& $LM_{adj}$ & 0.9085 & 0.9800 & 0.9915 & 0.9980 & 0.9620 & 0.9345 & 0.9995 & 0.9980 & 0.9820 \\
		& CD & 0.4085 & 0.5100 & 0.5665 & 0.5940 & 0.4695 & 0.4475 & 0.5865 & 0.6085 & 0.5355 \\ \hline
		\multirow{4}{*}{100} & $LM_e$ & 0.9990 & 1 & 1 & 1 & 1 & 1 & 1 & 1 & 1  \\
		& PET & 1 & 1 & 1 & 1 & 1 & 1 & 1 & 1 & 1  \\
		& $LM_{adj}$ & 0.9990 & 1 & 1 & 1 & 1 &1  & 1 & 1 &  \\
		& CD & 0.6405 & 0.7685 & 0.8020 & 0.6955 & 0.8475 & 0.7805 & 0.6490 & 0.6830 & 0.7915 \\ \hline 
	\end{tabular}
\end{table}

\section{A real data analysis}\label{sec:realdata}

We  apply our  two new tests $LM_e$ and PET  to the public health data sets
in \cite{liworking2020} for the investigation of  association between  air
pollution,  hypertension and blood pressure of people.  They used a unique
sample of elderly people in Nanjing (China)  containing 21 individuals
with 441 medical records in total with precise examination dates, that
is, $n=21, T=21$. Fixed effects panel data models are constructed to
evaluate the effect of the special particulate $PM_{2.5}$ (diameter
$<2.5~ \mu$m) on the hypertension (HY), systolic blood pressure (SBP)
and diastolic blood pressure (DBP), respectively, by controlling the individual fixed
effects $\delta_i, i=1,\ldots, n$. Temperature (temp) is also considered as another
control variable due to its positive effect on blood
pressure. Therefore,  three panel models are considered
\begin{itemize}
	\item Model 1: $HY=\beta_1 \log (PM_{2.5})+\beta_2 \log(temp)+\delta_i+\nu_1$;
	\item Model 2: $SBP=\beta_1 \log (PM_{2.5})+\beta_2 \log(temp)+\delta_i+\nu_2$;
	\item Model 3: $DBP=\beta_1 \log (PM_{2.5})+\beta_2 \log(temp)+\delta_i+\nu_3$.
\end{itemize}

To investigate whether the cross-sectional uncorrelation assumption in
three models is justified, we applied the $LM_{adj}$
test, the $LM_e$ test, the PET test and the CD test to each regression
model. The values of the corresponding test statistics are reported in
 Table~\ref{table9}.

 \begin{table}[ht]
	\centering
	\caption{Values of test statistics.}
	\label{table9}
	\begin{tabular}{|c|ccc|}
		\hline
		Tests & Model 1 & Model 2 & Model 3\\
		\hline 
		PET & -0.2035 & 31.56 & 158.3 \\
		$LM_e$ & 0.0694 & 12.75 & 34.65 \\
		$LM_{adj}$ & 0.9143 & 15.66 & 40.94 \\
		CD & 0.2049 & 16.71 & 28.35\\
		\hline 
	\end{tabular}
 \end{table}
 
  Among the
 other four tests and considering a 5\% nominal level,  none of them  can 
 reject the null
 hypothesis of cross-sectional uncorrelation under Model 1.  But they
 all  reject
 the null under both Model 2 and Model 3. That means, individuals are
 uncorrelated in terms of hypertension, but correlated in terms of SBP
 and DBP. Meanwhile, the values of PET are much larger than others in
 Models 2 and  3 which confirm  its power-enhancement
 property.
 Overall, the statistical results drawn from Model 1 seem 
 reliable and consistent with each other,  while more control variables need to be investigated to
 study the association of air pollution with SBP or DBP.  

 As illustrated by this example, the results of this paper seem to have potential application in the important area of panel data modelling.

\section{Conclusion}\label{sec:concl}

For large fixed effects panel data model, we propose two new and
efficient tests to detect the existence of cross-sectional correlation
(dependence). The asymptotic normalities of test statistics are
constructed under a simultaneous limit scheme (SIM-L) where  the  cross-sectional units
dimension $n$ and  the time series dimension $T$ are both large with
comparable magnitude. 
Meanwhile, these results do not
need the normality  assumption on the errors and/or on the random
design, while such normality assumptions   are   essential for
  the theoretical justification of most of the existing tests for
large panels.  Extensive Monte-Carlo experiments demonstrates the
superiority of our proposed tests over some popular existing methods
in terms of size and power. Especially, the power enhanced
high-dimensional test PET consistently outperforms all the other
methods considered.

There are still several avenues for future research. The tests
proposed here are based on the fixed effects panel data model.
It is highly valuable to investigate their validity in  other large
panel  models. For example, our simulation experiments  have shown the
applicability of these proposed tests for the  heterogeneous panel
data model.  A thorough  theoretical investigation of such observation
is missing though.

\bibliographystyle{plainnat}
\bibliography{reference}

\appendix
\section{Proof of the asymptotic normality~\eqref{trace_error}}
\label{app:trace1}

From the Theorem 3.1 of \cite{zheng2020}, we have
\[tr(R_T^2) - \mu_c \stackrel{D}{\longrightarrow} N(\mu_{limit}, \sigma_{limit}^2). \]
From the result for $g_l=x^2$ in Example 3.2 of \cite{zheng2020}, we get the centring term
\[\mu_c = n\left(1+\frac{n}{T-1}\right) = n \left(1+\frac{n}{T}\left(1+\frac{1}{T-1} \right) \right)=n(1+c_T)+c_T^2. \]
From the results for $R=\mathbf{I}_n$ and $g_l=x^2$ in Example 3.3 of \cite{zheng2020}, we get the limiting terms
\[\mu_{limit}=-c, \ \textrm{and} \ \sigma_{limit}^2=4c^2. \]
The proof of Lemma~\ref{trace_error} is complete.
\section{Proof of the asymptotic normality~\eqref{PET_error}}
\label{sec:PET}
From the Theorem 3.2 of \cite{zheng2020}, we have
\[tr(R_T^4)-\mu_c^\prime \stackrel{D}{\longrightarrow} N(\mu_l^\prime, \sigma_l^{\prime2}). \]
From the result for $g_l=x^4$ in Example 3.2 of \cite{zheng2020}, we get the centring term
\[\mu_c^\prime = n\left(1+\frac{6n}{T-1}+\frac{6n^2}{(T-1)^2}+\frac{n^3}{(T-1)^3} \right). \]
From the results for $R=\mathbb{I}_n$ and $g_l=x^4$ in Example 3.3 of \cite{zheng2020}, we get the limiting terms
\[\mu_l^\prime = -6c(1+c)^2,  \]
and
\[\sigma_l^{\prime2} = 8c^4+96c^3(1+c)^2+16c^2(3c^2+8c+3)^2. \]
The proof of Lemma~\ref{PET_error} is complete.

\section{Proof of the key estimates \eqref{trace_residual} and \eqref{trace4_diff}}
\label{sec:trace4}

The proof for the two key estimates is 
the main technical difficulty of the paper.
They are given in Propositions~\ref{trace2_prop} and \ref{PET_prop}
at the end of this section after a series of preliminary lemmas and
calculations.

Recall first some useful notations related to the OLS residuals as
given at the beginning of Section~\ref{sec:OLS}.
The OLS estimator for $\beta$ is
$\hab=\left(XX^{\prime}\right)^{-1}XY$ ($k\times1$ vector).
For the errors  $\vv_{it}$, let $\VV_i=\begin{pmatrix} \vv_{i1},
\cdots, \vv_{iT} \end{pmatrix}^{\prime}$ be a ${T\times1}$ vector and
then $\VV=\begin{pmatrix} \VV_{1} , \cdots,  \VV_{n} \end{pmatrix}$ is
a $T\times n$ matrix. Similarly for the residuals,
define
$\vvv_{it}=\vv_{it}-{\tilde x_{it}}^{\prime}\left(\hab-\beta\right)$,
$\VVV_i=\VV_i-X_i^{\prime}\left(\hab-\beta\right)$ and  set
$\VVV=\begin{pmatrix} \VVV_{1} , \cdots,  \VVV_{n} \end{pmatrix}$
($T\times n$ matrix).
Define also $W_i=X_i^{\prime}\left(\hab-\beta\right)$ ($T\times1$ vector), $W=\begin{pmatrix} W_{1} , \cdots,  W_{n} \end{pmatrix}$ ($T\times n$ matrix). We can easily get that $\VVV_i=\VV_i-W_i$, $\VVV=\VV-W$.

We have for the sample covariance matrices:
$S_T=\frac1T\VV^{\prime}\VV$, $\hat{S}_T=\frac1T\VVV^{\prime}\VVV$
with respective elements,
$$\begin{aligned}
S_{T,i,j}&=\frac1T \inn{\VV_i}{\VV_j} =\frac1T\sum\limits_{t=1}^{T}\vv_{it}\vv_{jt},\\
\hat{S}_{T,i,j}&=\frac1T \inn{\VVV_i}{\VVV_j}=\frac1T\sum\limits_{t=1}^{T}\vvv_{it}\vvv_{jt}.
\end{aligned}
$$

\newcommand\slognn{\left(\{ \log n/n\}^{\frac12} \right)}
\newcommand\slognt{\left(\sqrt{\frac{\log n}{T}}\right)}
\newcommand\pa[1]{\left(#1\right)}

\begin{lemma}~\citep[Theorem 13 of Chapter 13]{petrov1975independent}
  \label{lem:petrov}
  Let $Y_1, \cdots, Y_n$ be independent and identically distributed random variables, such that $E(Y_1)=0$, $E(Y_1)^2=1$ and $E|Y_1|^r<\infty$ for some $r\geq3$. Then
	$$\left|\mathbb{P}(\frac{1}{\sqrt{n}}\sum_{i=1}^nY_i<y)-\Phi(y)\right|\leq \frac{C(r)}{(1+|y|)^r}\left[\frac{E|Y_1|^3}{\sqrt{n}}+ \frac{E|Y_1|^r}{n^{\frac{r-2}{2}}}\right]$$ for all $y$, where $\Phi(\cdot)$ is the cumulative distribution function of standard normal  and $C(r)$ is a positive constant depending only on $r$.
\end{lemma}

\begin{lemma}\label{lem:max_order}
	Let $\{Y_{ij}\}_{i\geq1,j\geq 1}$ be an array of independent and  identically distributed random variables such that $E(Y_{11})=0$, $E(Y_{11})^2=1$ and $E|Y_{11}|^r<\infty$ for some $r\geq3$. Let $X_{in}=\frac{1}{\sqrt{n}}\sum_{j=1}^nY_{ij}$, then for any  $\epsilon >0$, we have 
	$$\max_{1\le i \le n}|X_{in}|=O_p(n^{\frac{1}{2r}+\epsilon}).$$
\end{lemma}
\begin{proof}
	For some $\alpha>0,$ let $c=\frac{1}{2r}+\epsilon$, we have 
	\begin{align*}
	\mathbb{P}(\max_{1\le i \le n}|X_{in}|\geq \alpha n^c) & = 1-\left[ \mathbb{P}(|X_{1n}|<\alpha n^c) \right]^n\\
	& \le 1- \{\Phi( \alpha n^c) -\frac{C(r)}{(1+ \alpha n^c)^r} \left[\frac{E|Y_{11}|^3}{\sqrt{n}}+ \frac{E|Y_{11}|^r}{n^{\frac{r-2}{2}}}\right] \}^n \\
	&\sim 1- \{1-\frac{1}{\sqrt{2\pi}\alpha}n^{-c}e^{-\frac{\alpha^2n^{2c}}{2}}
	-\frac{C(r)}{(1+ \alpha n^c)^r} \left[\frac{E|Y_{11}|^3}{\sqrt{n}}+ \frac{E|Y_{11}|^r}{n^{\frac{r-2}{2}}}\right] \}^n \\
	&\sim \frac{1}{\sqrt{2\pi}\alpha}n^{1-c}e^{-\frac{\alpha^2n^{2c}}{2}}+ \frac{C(r)E|Y_{11}|^3\sqrt{n}}{(1+ \alpha n^c)^r}+ \frac{C(r)E|Y_{11}|^r}{(1+ \alpha n^c)^rn^{\frac{r}{2}-2}},
	\end{align*}
	where the first inequality follows by Lemma \ref{lem:petrov}, and the first approximation follows by the fact that $\Phi(x) \sim 1- \frac{1}{\sqrt{2\pi}x}e^{-\frac{x^2}{2}}$ for large $x$. Therefore, $\max_{1\le i \le n}|X_{in}|=O_p(n^c)$ holds.
\end{proof}

\begin{remark}
	For the panel data model, if the errors
    $\{\tilde{\nu}_{it}\}_{i\geq1,t\geq 1}$
    satisfy the conditions in Lemma~\ref{lem:max_order} and
    $X_{iT}=\frac{1}{\sqrt{T}}\sum_{t=1}^T\tilde{\nu}_{it}$, then for any
    $\epsilon >0$, the estimate
    $\max_{1\le i \le n}|X_{iT}|=O_p(n^{\frac{1}{2r}+\epsilon})$ still
    holds once we further assume that $K_1\leq\frac{n}{T}\leq K_2$ for
    some positive constants $K_1$ and $K_2$. The latter holds in the
    SIM-L scheme where indeed  $\frac{n}{T}\rightarrow c>0$.
\end{remark}

\begin{lemma}\label{lem:estimates}
	Under  
	Assumptions (B1) and (B2),
    for any $\epsilon>0$ and some integer $r_1\geq 6$ (such that
    $E\nu_{it}^6<\infty$), we have the following results,
	\begin{enumerate}
		\item[(a)\ ]  $\displaystyle \max_{1\le i ,  j \le n}    |\inn{\VV_i}{W_j}|=O_p (n^{\frac{1}{2r_1}+\epsilon-1/2})$.
		\item[(b)\ ]  $\displaystyle \max_{1\le i, j \le n}    |\inn{W_i}{W_j}|=O_p (n^{-1})$.
		
		\item[(c)\ ]  $\displaystyle \max_{1\le i  \le n} |S_{T,i,i} -\sigma^2|=O_p (n^{\frac{1}{r_1}+\epsilon-1/2})$.
		\item[(d)\ ]
		$\displaystyle   \max_{1\le i  \le n} |   \hat S_{T,i,i} -    S_{T,i,i} |   =O_p \pa{   n^{\frac{1}{2r_1}+\epsilon-3/2} }.$
		\item[(e)\ ]  $\displaystyle \max_{1\le i  \le n} |\hat S_{T,i,i} -\sigma^2|= O_p \left(n^{\frac{1}{2r_1}+\epsilon-3/2}\right)$.
		\item[(f)\ ]  $\displaystyle   \max_{1\le i  \le n} |   \hat S_{T,i,i}^2 -    S_{T,i,i}^2 |   =O_p \left(n^{\frac{1}{2r_1}+\epsilon-3/2} \right).$
		
	\end{enumerate}
\end{lemma}

\begin{proof}
	Throughtout the proof, the linear operators $X_iX_i'$ is
    restricted to the subspace $E_k$ (orthogonal to the constant
    vectors in $R^k$) where it is invertible by Assumption (B2). Note that $\epsilon_i$ for $i=1,2$ are small positive numbers which may vary in different equations.

	(a). Firstly, we consider the case $i=j$. By CLT, we have $\sqrt
    T\left(\frac1T\VV_i^\prime\VV_i-1\right)\cvd N(0,\tau^2)$, where
    $\tau^2=var(\vv_{it}^2)$. By Assumption (B2), we have for any
    non null $u\in E_k$, 
	$$
	\frac{\inn{u}{\sum\limits_{t}{\tilde x_{it}}\vv_{it}}}{  \left( u' X_i X_i' u\right)^{\frac12}}\cvd N\left(0,1\right).
	$$ 
	That is 
	$$\xi_i = \left( X_i X_i' \right)^{-\frac12} \sum\limits_{t}{\tilde x_{it}}\vv_{it}\cvd N_k\left(0,I_k\right).
	$$
	Therefore,  for some $r_1 \geq 3$ and for any $\epsilon_1>0$, from Lemma \ref{lem:max_order} we have
	\[ \max_{1\le i\le n} \|\xi_i\| =O_p(n^{\frac{1}{2r_1}+\epsilon_1}).
	\]
	Then
	\begin{align*}
	\max_{1\le i \le n} \left|\inn{W_i}{\VV_i}\right|=\max_{1\le i  \le n}\left|\sum\limits_t \vv_{it}{\tilde x_{it}}^\prime\left(\hab-\beta\right)\right|
	&= \max_{1\le i  \le n}\left|\sqrt T\xi_i^\prime\left(\frac1TX_iX_i^\prime\right)^{\frac12}\left(\hab-\beta\right)\right|\\
	&\leq \sqrt T\left|\left|\left(\frac1TX_iX_i^\prime\right)^{\frac12}\right|\right|\cdot\max_{1\le i  \le n}\left\|\xi_i\right\|\cdot\left\|\hab-\beta\right\|\\
	&=O_p\left(\sqrt T\cdot\frac{n^{\frac{1}{2r_1}+\epsilon_1}}{\sqrt {nT}}\right)=O_p(n^{\frac{1}{2r_1}+\epsilon_1-1/2}).
	\end{align*}
	The calculations for $i\ne j$ case is similar.
	
	\medskip
	
	(b). By Assumption (B2), we have
	\begin{align*}
	\left| \inn{W_i}{W_j} \right| &\le   \| W_i\| \| W_j\|  = \| X_i'    (\hab-\beta)\| \| X_j'    (\hab-\beta)\|
	 \le  \| X_i\| \|X_j\| \|\hab-\beta\|^2 \\
	& \le  \| X_i'X_i\|^{1/2} \|X_j'X_j\|^{1/2} \|\hab-\beta\|^2 \\
	& =  O(T^{1/2})  O(T^{1/2}) O_p(1/(nT))  = O_p(n^{-1}).
	\end{align*}
	
	\bigskip
	
	(c). Note for $r_2=r_1/2$, we have    $E|\tilde{\nu}_{it}^2-\sigma^2|^{r_2}<\infty$.
    By CLT, we have $Z_i:= \sqrt T (S_{T,i,i}-\sigma^2)\cvd
	N(0,\tau^2)$ where $\sigma^2=E(\vv_{it}^2)$ and $\tau^2=\text{var}
	(\vv_{it}^2)$, then we have
	$ \max\limits_{1\le i\le n} |Z_i| =    O_p(n^{\frac{1}{2r_2}+\epsilon_2-1/2})  $
    for any $\epsilon_2>0$ by Lemma \ref{lem:max_order}.
    Therefore,
	\[  \max\limits_{1\le i\le n} |S_{T,i,i}-\sigma^2| = O_p(n^{\frac{1}{2r_2}+\epsilon_2-1/2})=O_p(n^{\frac{1}{r_1}+\epsilon_2-1/2}).
	\]
	
	\medskip
	
	(d). By (a) and (b), 
	we have 
	\begin{align*}
	\nonumber
	\max_{1\le i  \le n} | \hat S_{T,i,i} -  S_{T,i,i} | & =\max_{1\le i \le n} \left|  -2T^{-1} \inn{W_i}{\VV_i} +T^{-1}\|W_i\|^2     \right|\\
	& = \frac1T  O_p (n^{\frac{1}{2r_1}+\epsilon_1-1/2}) +  \frac1T  O_p (n^{-1}) = O_p
	(  T^{-1}n^{\frac{1}{2r_1}+\epsilon_1-1/2}  ) . 
	\end{align*}
	
	\medskip
    
	(e).  The conclusion holds from (c) and (d). 
	
	\medskip
	
	(f). By (a), (b) and (c), we have
	\begin{eqnarray*}
      \lefteqn{	\max_{1\le i  \le n}| \hat S_{T,i,i}^2 -  S_{T,i,i}^2        |} \\
      & =& \max_{1\le i \le n} \frac{1}{T^2}\left|
	  4\inn{W_i}{\VV_i}^2 + \|W_i\|^4-4\inn{\VV_i}{\VV_i} \inn{W_i}{\VV_i} +2\inn{\VV_i}{\VV_i}\|W_i\|^2 +4 \inn{W_i}{\VV_i} \|W_i\|^2 \right|\\
	  &=&  \frac{1}{T^2}\left(O_p\left(n^{\frac{1}{r_1}+2\epsilon_1-1} \right) +O_p\left(\frac{1}{n^2} \right)+O_p\left( n^{\frac{1}{2r_1}+\epsilon_1+1/2}\right) + O_p\left(1 \right) +  O_p\left(n^{\frac{1}{2r_1}+\epsilon_1-3/2} \right) \right)\\
	  &=&O_p \left( T^{-2}n^{\frac{1}{2r_1}+\epsilon_1+1/2} \right).
    \end{eqnarray*}	
\end{proof}

\begin{lemma}\label{lem:trace2}
	Suppose Assumptions (B1)-(B2) hold.  We have
	\begin{enumerate}
		\item[(a)\ ]  $ \sum\limits_{i\not=j}^n\left\{\hat{S}_{T,i,j}^2-S_{T,i,j}^2\right\}\cvp0$.
		\item[(b)\ ]  $ \sum\limits_{i\not=j\not=l}^n\left\{\hat{S}_{T,i,j}^2\hat{S}_{T,j,l}^2-S_{T,i,j}^2S_{T,j,l}^2\right\}\cvp0$.
		\item[(c)\ ] $ \sum\limits_{i\not=j\not=l\not=s}^n\left\{\hat{S}_{T,i,j}\hat{S}_{T,j,l}\hat{S}_{T,i,s}\hat{S}_{T,l,s}-S_{T,i,j}S_{T,j,l}S_{T,i,s}S_{T,l,s}\right\}\cvp0$.
	\end{enumerate}
\end{lemma}

\begin{proof}	
	(a) We have 
	
	$$\begin{aligned}
	\sum\limits_{i\not=j}^n \left\{\SSS_{T,i,j}^2-S_{T,i,j}^2\right\}
	&=\frac1{T^2}\sum\limits_{i\not=j}^n\left\{\left((\VV_i^\prime-W_i^\prime)(\VV_j-W_j)\right)^2-\left(\VV_i^\prime\VV_j\right)^2\right\} \\
	&=\frac1{T^2}\sum\limits_{i\not=j}^n\left\{-2\VV_i^\prime\VV_j\left(W_i^\prime\VV_j+\VV_i^\prime W_j-W_i^\prime W_j\right)+\left(W_i^\prime\VV_j+\VV_i^\prime W_j-W_i^\prime W_j\right)^2\right\} .\\
	\end{aligned}
	$$
	Furthermore,
	$$\begin{aligned}
	\frac1{T^2}\sum\limits_{i\not=j}^n{|\VV_i^\prime\VV_jW_i^\prime W_j|}&=\frac1{T^2}\sum\limits_{i\not=j}^n{|\sqrt T\xi_{ij}W_i^\prime W_j|}\ \ \text{set}\ \VV_i^\prime\VV_j=\sqrt T\cdot\xi_{ij}\   (\xi_{ij}\text{ tends to}\ N(0,1))\\
	&\leq \frac1{T^2}\left(\sum\limits_{i\not=j}^n{|\sqrt T |\cdot|W_i^\prime W_j|}\right)\cdot \max\limits_{i,j}|\xi_{ij}|\\
	&=O_p\left(\frac{n^2\sqrt{T}  n^{\frac{1}{2r_1}+\epsilon}}{T^2n}\right)
    =O_p\left(n^{\frac{1}{2r_1}+\epsilon-1/2} \right).\\
	\end{aligned}$$
	After similar calculations, we conclude that
	$$\begin{aligned}
	\frac1{T^2}\sum\limits_{i\not=j}^n{|W_i^\prime\VV_jW_i^\prime \VV_j|}&=O_p\left(n^{\frac{1}{r_1}+\epsilon-1}\right),\\
	\frac1{T^2} \sum\limits_{i\not=j}^n{|W_i^\prime\VV_jW_i^\prime W_j|}&=O_p\left(n^{\frac{1}{2r_1}+\epsilon-3/2}\right),\\
	\frac1{T^2}\sum\limits_{i\not=j}^n{|W_i^\prime W_jW_i^\prime W_j|}&=O_p\left(n^{-2}\right).\\
	\end{aligned}$$
	
	\bigskip

	The last result we need to show is 
	\[
	M_n :
	= \sum\limits_{i\not=j}^n\frac{\VV_i^\prime\VV_jW_i^\prime
		\VV_j+\VV_i^\prime\VV_j\VV_i^\prime W_j}{T^2}  =o_p(1).
	\]
	Consider the $U$-statistic
	\[  U_n = \frac1{\binom{n}{2}}\sum_{\{i,j\}} h_T(\VV_i,\VV_j),
	\quad h_T(\VV_i,\VV_j) =
	\frac{\VV_i^\prime\VV_j  \VV_j'X_i' 
		+\VV_i^\prime\VV_j\VV_i' X_j' }{T }.
	\]
	We have
	\begin{equation}
	\label{eq:Mn}
	M_n = \sum\limits_{i\not=j}^n\frac{\VV_i^\prime\VV_j  \VV_j'X_i' 
		+\VV_i^\prime\VV_j\VV_i' X_j' }{T^2} \cdot (\hat
	\beta-\beta)
	= \frac2T \binom{n}{2} U_n \cdot (\hat\beta-\beta).
	\end{equation}
	Because the dimension $k$ is fixed, we suppose in the following that
	$k=1$ to simplify the presentation. Thus $X_j'$ is a $T\times 1$
	vector. 
	By direct calculations, one can show that
	\[  E(\VV_i^\prime\VV_j  \VV_i'X_j')^2 =E(\vv_{11}^2)
	\left\{ E(\vv_{11}^4) + (T-1) [E(\vv_{11}^2)]^2\right\}
	\sum_{t=1}^Tx_{jt}^2.
	\]
	Therefore there exist positive constants $a_2>a_1>0$, such that
	\[
	a_1 T^2 \le       E(\VV_i^\prime\VV_j  \VV_i'X_j')^2\le a_2T^2.
	\]
	It follows that exist positive constants $b_2>b_1>0$, such that
	\begin{equation}
	\label{eq:moment2}
	b_1  \le       E \{ h_T^2(\VV_i, \VV_j)\} \le b_2.
	\end{equation}
	Similarly, one can show that
	\begin{equation}
	\label{eq:moment1}
	E \{ \sqrt T h_T (\VV_i, \VV_j)\}  =  o(1).
	\end{equation}
	By  CLT for $U$-statistic (e.g. Theorem 12.3 in \cite{van1998asymptotic}), we find that 
	$\sqrt n  U_n $ is asymptotic Gaussian.
	Finally, we have by \eqref{eq:Mn}-\eqref{eq:moment2}-\eqref{eq:moment1}
	\[
	M_n = \frac2T \binom{n}{2} U_n \cdot (\hat\beta-\beta)
	=\frac2T \binom{n}{2} \cdot \frac1{\sqrt n} O_p(1)  \cdot O_p(n^{-1})
	= O_p(1/\sqrt{n}).
	\]
	The proof of the (a) is   complete.
	
	\bigskip
	
	(b) We have 
	\begin{align*}
	LHS
	&= \frac1{T^4}\sum\limits_{i\neq j\neq l}^n\left\{ {\left((\VV_i^\prime-W_i^\prime)(\VV_j-W_j)\right)^2\left((\VV_j^\prime-W_j^\prime)(\VV_l-W_l)\right)^2 -\left(\VV_i^\prime\VV_j\right)^2\left(\VV_j^\prime\VV_l\right)^2}\right\}\\
	&= \frac1{T^4}\sum\limits_{i\neq j\neq l}^n \left\{{\left(\VV_i^\prime\VV_j\right)^2\left(\VV_j^\prime W_l+W_j^\prime \VV_l-W_j^\prime W_l \right)^2+\left(\VV_j^\prime\VV_l\right)^2\left(\VV_i^\prime W_j+W_i^\prime \VV_j-W_i^\prime W_j \right)^2}\right\}\\
	& +\frac1{T^4} \sum\limits_{i\neq j\neq l}^n {\left(\VV_j^\prime W_l+W_j^\prime \VV_l-W_j^\prime W_l \right)^2\left(\VV_i^\prime W_j+W_i^\prime \VV_j-W_i^\prime W_j \right)^2}\\
	& -2\frac1{T^4} \sum\limits_{i\neq j\neq l}^n {\left(\VV_i^\prime\VV_j\right)\left(\VV_j^\prime W_l+W_j^\prime \VV_l-W_j^\prime W_l \right)^2\left(\VV_i^\prime W_j+W_i^\prime \VV_j-W_i^\prime W_j \right)}\\
	& -2 \frac1{T^4}\sum\limits_{i\neq j\neq l}^n {\left(\VV_j^\prime\VV_l\right) \left(\VV_i^\prime W_j+W_i^\prime \VV_j-W_i^\prime W_j \right)^2\left(\VV_j^\prime W_l+W_j^\prime \VV_l-W_j^\prime W_l \right)}\\
	& +4\frac1{T^4}\sum\limits_{i\neq j\neq l}^n {\left(\VV_i^\prime\VV_j\right)\left(\VV_j^\prime W_l+W_j^\prime \VV_l-W_j^\prime W_l \right)\left(\VV_j^\prime\VV_l\right) \left(\VV_i^\prime W_j+W_i^\prime \VV_j-W_i^\prime W_j \right)}\\
	& -2\frac1{T^4} \sum\limits_{i\neq j\neq l}^n{\left(\VV_i^\prime\VV_j\right)^2 \left(\VV_j^\prime\VV_l\right) \left(\VV_j^\prime W_l+W_j^\prime \VV_l-W_j^\prime W_l \right)}\\
	& -2 \frac1{T^4}\sum\limits_{i\neq j\neq l}^n{\left(\VV_j^\prime\VV_l\right)^2 \left(\VV_i^\prime\VV_j\right) \left(\VV_i^\prime W_j+W_i^\prime \VV_j-W_i^\prime W_j \right)}.
    \end{align*}
	It is easy to conclude that
	\begin{align*}
	&\frac1{T^4}\sum\limits_{i\neq j\neq l}^n {\left|\VV_i^\prime \VV_j \VV_i^\prime \VV_j \VV_j^\prime W_l\VV_j^\prime W_l \right|} = O_p\left(n^{\frac{2}{r_1}+\epsilon-1} \right),\\
	&\frac1{T^4}\sum\limits_{i\neq j\neq l}^n {\left|\VV_i^\prime \VV_j \VV_i^\prime \VV_j W_j^\prime W_lW_j^\prime W_l \right|} = O_p\left(n^{\frac{1}{r_1}+\epsilon-2} \right),\\
	& \frac1{T^4}\sum\limits_{i\neq j\neq l}^n{\left|\VV_i^\prime \VV_j \VV_i^\prime \VV_j  \VV_j^\prime W_lW_j^\prime W_l \right|} = O_p\left(n^{\frac{3}{2r_1}+\epsilon-3/2} \right),\\
	&\frac1{T^4} \sum\limits_{i\neq j\neq l}^n {\left|\VV_j^\prime W_l \VV_j^\prime W_l \VV_i^\prime W_j \VV_i^\prime W_j \right|} = O_p\left(n^{\frac{2}{r_1}+\epsilon-3} \right),\\
	& \frac1{T^4}\sum\limits_{i\neq j\neq l}^n {\left|\VV_j^\prime W_l W_j^\prime W_l \VV_i^\prime W_j W_i^\prime W_j \right|} = O_p\left(n^{\frac{1}{r_1}+\epsilon-4} \right),\\
	& \frac1{T^4}\sum\limits_{i\neq j\neq l}^n {\left|
		W_j^\prime W_l W_j^\prime W_l W_i^\prime W_j W_i^\prime W_j \right|} = O_p\left(n^{-5} \right),\\
	&\frac1{T^4}\sum\limits_{i\neq j\neq l}^n {\left|
		\VV_i^\prime \VV_j \VV_j^\prime W_l\VV_j^\prime W_l \VV_i^\prime W_j \right|} = O_p\left(n^{\frac{2}{r_1}+\epsilon-2} \right),\\
	&\frac1{T^4}\sum\limits_{i\neq j\neq l}^n {\left|
		\VV_i^\prime \VV_j \VV_j^\prime W_l\VV_j^\prime W_l W_i^\prime W_j \right|} = O_p\left(n^{\frac{3}{2r_1}+\epsilon-5/2} \right),\\
	&\frac1{T^4}\sum\limits_{i\neq j\neq l}^n{\left|
		\VV_i^\prime \VV_j \VV_j^\prime W_lW_j^\prime W_l W_i^\prime W_j \right|} = O_p\left(n^{\frac{1}{r_1}+\epsilon-3} \right),\\
	&\frac1{T^4}\sum\limits_{i\neq j\neq l}^n {\left|
		\VV_i^\prime \VV_j W_j^\prime W_lW_j^\prime W_l W_i^\prime W_j \right|} = O_p\left(n^{\frac{1}{2r_1}+\epsilon-7/2} \right),\\
	&\frac1{T^4}\sum\limits_{i\neq j\neq l}^n {\left|
		\VV_i^\prime \VV_j\VV_i^\prime \VV_j\VV_j^\prime \VV_l W_i^\prime W_j \right|} = O_p\left(n^{\frac{3}{2r_1}+\epsilon-1/2} \right). 
	\end{align*}
	The last result we need to show is
	\[L_n:=\sum\limits_{i\neq j\neq l}^n \frac{
		\VV_i^\prime \VV_j\VV_i^\prime \VV_j\VV_j^\prime \VV_l \VV_j^\prime W_l + \VV_i^\prime \VV_j\VV_i^\prime \VV_j\VV_j^\prime \VV_l W_j^\prime \VV_l }{T^4} = o_p(1). \]
	Consider the $U$-statistic
	\[U_n = \frac1{\binom{n}{3}}\sum_{\{i,j,l\}} h_T(\VV_i,\VV_j,\VV_l),\quad h_T(\VV_i,\VV_j,\VV_l)=\frac{\VV_i^\prime \VV_j\VV_i^\prime \VV_j\VV_j^\prime \VV_l \VV_j^\prime X_l^\prime + \VV_i^\prime \VV_j\VV_i^\prime \VV_j\VV_j^\prime \VV_l \VV_l^\prime X_j^\prime }{T^2}. \]
	We have
	\[L_n = \frac{6}{T^2}\binom{n}{3}U_n\cdot (\hat{\beta}-\beta). \]
	One can show that
	$$\begin{aligned}
	&E\left(\VV_i^\prime \VV_j\VV_i^\prime \VV_j\VV_j^\prime \VV_l \VV_j^\prime X_l^\prime \right)^2 \\
	&= \sum_{t=1}^T x_{lt}^2 \cdot \Big(E(\vv_{11}^2)E(\vv_{11}^4)E(\vv_{11}^8)+ (T-1)E(\vv_{11}^2)E^3(\vv_{11}^4) +(T-1)E^2(\vv_{11}^2)E(\vv_{11}^4)E(\vv_{11}^6)\\
	&+(T-1)(T-2)E^3(\vv_{11}^2)E^2(\vv_{11}^4) +3(T-1)E^4(\vv_{11}^2)E(\vv_{11}^6)+2(T-1)(T-2)E^5(\vv_{11}^2)E(\vv_{11}^4)\\
	&+(T-1)E^3(\vv_{11}^2)E^2(\vv_{11}^4)+(T-1)(T-2)(T-3)E^7(\vv_{11}^2) \Big).
	\end{aligned}$$
	Therefore there exist positive constants $a_2>a_1>0$, such that
	\[a_1 T^4 \leq E\left(\VV_i^\prime \VV_j\VV_i^\prime \VV_j\VV_j^\prime \VV_l \VV_j^\prime X_l^\prime \right)^2 \leq a_2 T^4.  \]
	It follows that exist positive constants $b_2>b_1>0$, such that
	\[b_1 \leq E\left\{h_T^2(\VV_i,\VV_j,\VV_l) \right\} \leq b_2. \]
	One can show that
	\[E\{\sqrt{T}h_T(\VV_i,\VV_j,\VV_l) \}=o(1). \]
	By CLT for $U$-statistic, we find that $\sqrt{n}U_n$ is asymptotic Gaussian. Finally, we have 
	\[L_n=\frac{6}{T^2}\binom{n}{3}U_n\cdot (\hat{\beta}-\beta)=\frac{6}{T^2}\binom{n}{3} \cdot \frac{1}{\sqrt{n}}\cdot O_p\left(\frac{1}{n} \right)=O_p\left(\frac{1}{\sqrt{n}} \right). \]
	The proof of (b) is completed.
	
	(c) We have
	\begin{align*}
	LHS&=\sum\limits_{i\neq j\neq l\neq s}^n \left\{{\SSS_{T,i,j}\SSS_{T,j,l}\SSS_{T,i,s}\SSS_{T,s,l}-S_{T,i,j}S_{T,j,l}S_{T,i,s}S_{T,s,l}}\right\}\\
	&= \frac1{T^4}\sum\limits_{i\neq j\neq l\neq s}^n {(\VV_i^\prime-W_i^\prime)(\VV_j-W_j)(\VV_j^\prime-W_j^\prime)(\VV_l-W_l) (\VV_i^\prime-W_i^\prime)(\VV_s-W_s)(\VV_s^\prime-W_s^\prime)(\VV_l-W_l)}\\
	& - \frac1{T^4}\sum\limits_{i\neq j\neq l\neq s}^n {\VV_i^\prime \VV_j \VV_j^\prime \VV_l \VV_i^\prime \VV_s \VV_s^\prime \VV_l}\\
	&= \frac1{T^4}\sum\limits_{i\neq j\neq l\neq s}^n
	\Bigg((\VV_i^\prime \VV_j \VV_j^\prime \VV_l)\Big\{-\VV_i^\prime \VV_s W_s^\prime \VV_l-\VV_i^\prime \VV_s \VV_s^\prime W_l +\VV_i^\prime \VV_s W_s^\prime W_l-W_i^\prime \VV_s \VV_s^\prime \VV_l +W_i^\prime \VV_s W_s^\prime \VV_l \\
	& +W_i^\prime \VV_s\VV_s\prime W_l-W_i^\prime \VV_s W_s^\prime W_l- \VV_i^\prime W_s \VV_s^\prime \VV_l +W_i^\prime \VV_s W_s^\prime \VV_l +W_i^\prime \VV_s \VV_s^\prime W_l -W_i^\prime \VV_s W_s^\prime W_l + W_i^\prime W_s \VV_s^\prime \VV_l\\
	&  -W_i^\prime W_s W_s^\prime \VV_l -W_i^\prime W_s \VV_s^\prime    W_l+ W_i^\prime W_s W_s^\prime W_l \Big\}\\
    &+\Big\{-\VV_i^\prime \VV_j W_j^\prime \VV_l-\VV_i^\prime \VV_j \VV_j^\prime W_l +\VV_i^\prime \VV_j W_j^\prime W_l-W_i^\prime \VV_j \VV_j^\prime \VV_l +W_i^\prime \VV_j W_j^\prime \VV_l +W_i^\prime \VV_j\VV_j\prime W_l-W_i^\prime \VV_j W_j^\prime W_l\\
	& - \VV_i^\prime W_j \VV_j^\prime \VV_l +W_i^\prime \VV_j W_j^\prime \VV_l +W_i^\prime \VV_j \VV_j^\prime W_l -W_i^\prime \VV_j W_j^\prime W_l + W_i^\prime W_j \VV_j^\prime \VV_l -W_i^\prime W_j W_j^\prime \VV_l -W_i^\prime W_j \VV_j^\prime W_l\\
	& + W_i^\prime W_j W_j^\prime W_l \Big\} \cdot \Big\{-\VV_i^\prime \VV_s W_s^\prime \VV_l-\VV_i^\prime \VV_s \VV_s^\prime W_l +\VV_i^\prime \VV_s W_s^\prime W_l-W_i^\prime \VV_s \VV_s^\prime \VV_l +W_i^\prime \VV_s W_s^\prime \VV_l \\
	& +W_i^\prime \VV_s\VV_s\prime W_l-W_i^\prime \VV_s W_s^\prime W_l- \VV_i^\prime W_s \VV_s^\prime \VV_l +W_i^\prime \VV_s W_s^\prime \VV_l +W_i^\prime \VV_s \VV_s^\prime W_l -W_i^\prime \VV_s W_s^\prime W_l + W_i^\prime W_s \VV_s^\prime \VV_l\\
	&  -W_i^\prime W_s W_s^\prime \VV_l -W_i^\prime W_s \VV_s^\prime W_l+ W_i^\prime W_s W_s^\prime W_l \Big\}\Bigg).
	\end{align*}
	It is easy to conclude that 
	$$\begin{aligned}
	&\frac1{T^4}\sum\limits_{i\neq j\neq l\neq s}^n {\left|\VV_i^\prime \VV_j W_j^\prime \VV_l \VV_i^\prime \VV_s W_s^\prime W_l \right|}=O_p\left(n^{\frac{3}{2r_1}+\epsilon-1/2} \right),\\
	&\frac1{T^4} \sum\limits_{i\neq j\neq l\neq s}^n {\left|\VV_i^\prime \VV_j W_j^\prime \VV_l W_i^\prime \VV_s W_s^\prime \VV_l \right|}=O_p\left(n^{\frac{2}{r_1}+\epsilon-1} \right),\\
	& \frac1{T^4}\sum\limits_{i\neq j\neq l\neq s}^n {\left|\VV_i^\prime \VV_j W_j^\prime \VV_l W_i^\prime W_s W_s^\prime \VV_l \right|}=O_p\left(n^{\frac{3}{2r_1}+\epsilon-3/2}\right),\\
	&\frac1{T^4} \sum\limits_{i\neq j\neq l\neq s}^n {\left|\VV_i^\prime \VV_j W_j^\prime \VV_l W_i^\prime W_s W_s^\prime W_l \right|}=O_p\left(n^{\frac{1}{r_1}+\epsilon-2} \right),\\
	&\frac1{T^4}\sum\limits_{i\neq j\neq l\neq s}^n {\left|\VV_i^\prime \VV_j W_j^\prime W_l \VV_i^\prime \VV_s W_s^\prime W_l \right|}=O_p\left(n^{\frac{1}{r_1}+\epsilon-1} \right),\\
	&\frac1{T^4}\sum\limits_{i\neq j\neq l\neq s}^n {\left|\VV_i^\prime \VV_j W_j^\prime W_l W_i^\prime W_s W_s^\prime W_l \right|}=O_p\left(n^{\frac{1}{2r_1}+\epsilon-5/2} \right),\\
	&\frac1{T^4}\sum\limits_{i\neq j\neq l\neq s}^n{\left|W_i^\prime \VV_j W_j^\prime \VV_l W_i^\prime W_s W_s^\prime \VV_l \right|}=O_p\left(n^{\frac{3}{2r_1}+\epsilon-5/2} \right),\\
	&\frac1{T^4}\sum\limits_{i\neq j\neq l\neq s}^n {\left|W_i^\prime \VV_j W_j^\prime \VV_lW_i^\prime W_s W_s^\prime W_l \right|}=O_p\left(n^{\frac{1}{r_1}+\epsilon-3} \right),\\
	&\frac1{T^4}\sum\limits_{i\neq j\neq l\neq s}^n {\left|W_i^\prime W_j W_j^\prime \VV_l W_i^\prime W_s W_s^\prime W_l \right|}=O_p\left(n^{\frac{1}{2r_1}+\epsilon-7/2} \right)\\
	&\frac1{T^4}\sum\limits_{i\neq j\neq l\neq s}^n {\left|W_i^\prime W_j W_j^\prime W_l W_i^\prime W_s W_s^\prime W_l \right|}=O_p\left(n^{-4} \right).
	\end{aligned}$$
	We need to show that 
	$$\begin{aligned}
	&O_n:=\frac1{T^4}\sum\limits_{i\neq j\neq l\neq s}^n \VV_i^\prime \VV_j W_j^\prime \VV_l\VV_i^\prime \VV_s W_s^\prime \VV_l =o_p(1).\\
	& Q_n:=\frac1{T^4} \sum\limits_{i\neq j\neq l\neq s}^n {\VV_i^\prime \VV_j \VV_j^\prime \VV_l\VV_i^\prime \VV_s \VV_s^\prime W_l }=o_p(1).
	\end{aligned}$$
	Consider the $U$-statistic
	\[U_n = \frac1{\binom{n}{4}}\sum_{\{i,j,l,s\}} h_T(\VV_i,\VV_j,\VV_l,\VV_s),\quad h_T(\VV_i,\VV_j,\VV_l,\VV_s)=\frac{\VV_i^\prime \VV_j  \VV_l^\prime X_j^\prime \VV_i^\prime \VV_s  \VV_l^\prime X_s^\prime}{T^2}, \]
	so we have
	\[O_n = \frac{24}{T^2}\binom{n}{4}U_n\cdot (\hat{\beta}-\beta)^2. \]
	One can show that
	$$\begin{aligned}
	&E\left(\VV_i^\prime \VV_j  \VV_l^\prime X_j^\prime \VV_i^\prime \VV_s  \VV_l^\prime X_s^\prime \right)^2\\
	&= \sum_{t=1}^T x_{jt}^2 x_{st}^2 \cdot \Big(E^2(\vv_{11}^2)E^2(\vv_{11}^4)+(T-1)E^2(\vv_{11}^2)E^2(\vv_{11}^4)+2(T-1)E^4(\vv_{11}^2)E(\vv_{11}^4)\\
	&\hskip1cm +2(T-1)(T-2)E^4(\vv_{11}^2)E(\vv_{11}^4)+(T-1)E^6(\vv_{11}^2)+(T-1)(T-2)E^6(\vv_{11}^2)\\
	&\hskip1cm +(T-1)(T-2)(T-3)E^6(\vv_{11}^2) \Big).
	\end{aligned}$$
	Therefore there exist positive constants $a_2>a_1>0$, such that
	\[a_1 T^4\leq E\left(\VV_i^\prime \VV_j  \VV_l^\prime X_j^\prime \VV_i^\prime \VV_s  \VV_l^\prime X_s^\prime \right)^2 \leq a_2T^4. \]
	It follows that exist positive constants $b_2>b_2>0$, such that
	\[b_1\leq E\{h^2_T(\VV_i,\VV_j,\VV_l,\VV_s) \} \leq b_2. \]
	Similarly, one can show that
	\[E\{ \sqrt{T}h_T(\VV_i,\VV_j,\VV_l,\VV_s)\}=o(1). \]
	By CLT for $U$-statistic, we find that $\sqrt{n}U_n$ is asymptotic Gaussian. Finally, we have
	\[O_n = \frac{24}{T^2}\binom{n}{4}U_n\cdot (\hat{\beta}-\beta)^2= \frac{24}{T^2}\binom{n}{4}\cdot \frac{1}{\sqrt{n}}O_p(1)\cdot O_p\left(\frac{1}{n^2} \right)=O_p\left(\frac{1}{\sqrt{n}} \right). \]
	Similarly, one can prove that $Q_n=O_p\left(\frac{1}{\sqrt{n}} \right)$.
	The proof of (c) is completed.
\end{proof}
\begin{lemma} \label{lem:equal index}
	Suppose 
	Assumptions (B1)-(B2) hold.  We have
	\begin{enumerate}
		\item[(a)\ ]  $ \sum\limits_{i=1}^n\left\{\hat{S}_{T,i,i}^2-S_{T,i,i}^2\right\}\cvp0$.
		\item[(b)\ ]  $ \sum\limits_{i=1}^n\left\{\hat{S}_{T,i,i}^4-S_{T,i,i}^4\right\}\cvp0$.
	\end{enumerate}
\end{lemma}

\begin{proof}
	(a)
	Indeed, we have 
	\begin{align*}
	\left| \sum\limits_{i=1}^n \left\{  \hat{S}_{T,i,i}^2-S_{T,i,i}^2 \right\}\right| 
	& \le 
	\sum\limits_{i=1}^n 
	\left| \hat{S}_{T,i,i} +S_{T,i,i} \right|   \left| \hat{S}_{T,i,i} -S_{T,i,i} \right| \\
	&\le 
	\left(  \sum\limits_{i=1}^n   \left| \hat{S}_{T,i,i} +S_{T,i,i}  \right|^2\right)^{1/2}
	\left(  \sum\limits_{i=1}^n   \left| \hat{S}_{T,i,i} -S_{T,i,i}  \right|^2\right)^{1/2}\\
	& =   \left\{ O_p(n) \right\}^{1/2} \cdot
	\left\{  O_p  \pa{ n \frac1{T^2} n^{\frac{1}{r_1}+\epsilon-1} }   \right\}^{1/2}
	= O_p \pa{   n^{\frac{1}{2r_1}+\epsilon-1/2} }, 
	\end{align*}
	where the last equality comes from Lemma~\ref{lem:estimates}.
	The proof is completed.
	
	(b) We have 
	\begin{align*}
	\left| \sum\limits_{i=1}^n \left\{  \hat{S}_{T,i,i}^4-S_{T,i,i}^4 \right\}\right| 
	& \le 
	\sum\limits_{i=1}^n 
	\left| \hat{S}_{T,i,i}^2 +S_{T,i,i}^2 \right|   \left| \hat{S}_{T,i,i}^2 -S_{T,i,i} ^2\right| \\
	&\le 
	\left(  \sum\limits_{i=1}^n   \left| \hat{S}_{T,i,i}^2 +S_{T,i,i} ^2 \right|^2\right)^{1/2}
	\left(  \sum\limits_{i=1}^n   \left| \hat{S}_{T,i,i}^2 -S_{T,i,i} ^2 \right|^2\right)^{1/2}\\
	&
	= O_p \pa{    n^{\frac{1}{2r_1}+\epsilon-1/2}  }, 
	\end{align*}
	where the last equality comes from Lemma~\ref{lem:estimates}.
	The proof is completed.
\end{proof}

\begin{corollary} \label{cor:diff}
	Suppose Assumption 1-2 hold.  We have
	\begin{enumerate}
		\item[(a)\ ]  $ \sum\limits_{i,j}^n\left\{\hat{S}_{T,i,j}^2-S_{T,i,j}^2\right\}\cvp0$.
		\item[(b)\ ]  $ \sum\limits_{i,j,l}^n\left\{\hat{S}_{T,i,j}^2\hat{S}_{T,j,l}^2-S_{T,i,j}^2S_{T,j,l}^2\right\}\cvp0$.
		\item[(c)\ ]  $ \sum\limits_{i,j,l,s}^n\left\{\hat{S}_{T,i,j}\hat{S}_{T,j,l}\hat{S}_{T,i,s}\hat{S}_{T,l,s}-S_{T,i,j}S_{T,j,l}S_{T,i,s}S_{T,l,s}\right\}\cvp0$.
	\end{enumerate}
\end{corollary}
This corollary immediately holds with Lemma \ref{lem:trace2} and Lemma \ref{lem:equal index}.

\subsection{Main calculations}

\begin{proposition}\label{trace2_prop}
	Under Assumption (B1)-(B2),  we have
	\[ tr\left(\hat{R}_T^2-R_T^2\right) =o_p(1).
	\]
\end{proposition}

\begin{proof}
	It is easy to verify that
	\begin{equation}
	tr\left(\hat{R}_T^2-R_T^2\right)=\sum\limits_{i\not=j}^n  \left(  A_{ij} +B_{ij}\right) ,
	\end{equation}
	where
	\[
	A_{ij}=\frac{\SSS_{T,i,j}^2}{\SSS_{T,i,i}\SSS_{T,j,j}}-\frac{\SSS_{T,i,j}^2}{S_{T,i,i}S_{T,j,j}},\quad
	B_{ij}=\frac{\SSS_{T,i,j}^2-S_{T,i,j}^2}{S_{T,i,i}S_{T,j,j}}.
	\]
	We calculate $\sum\limits_{i\not=j}^n A_{ij}$ first. Note that
	\[A_{ij}=\SSS_{T,i,j}^2\left(\frac{S_{T,j,j}-\SSS_{T,j,j}}{\SSS_{T,j,j}S_{T,i,i}S_{T,j,j}}+\frac{S_{T,i,i}-\SSS_{T,i,i}}{\SSS_{T,i,i}\SSS_{T,j,j}S_{T,i,i}}
	\right).
	\]
	For the first term, using Lemma~\ref{lem:estimates} we have 
	$$\begin{aligned}
	\left|\sum\limits_{i\not=j}^n\SSS_{T,i,j}^2\frac{S_{T,j,j}-\SSS_{T,j,j}}{\SSS_{T,j,j}S_{T,i,i}S_{T,j,j}}\right|
	&=\left|\sum\limits_{i\not=j}^n\frac{\SSS_{T,i,j}^2}{\SSS_{T,j,j}S_{T,i,i}S_{T,j,j}}\cdot\left(\frac2T\VV_i^\prime W_i-\frac1TW_i^\prime W_i\right)\right|\\
	&\leq \sum\limits_{i\not=j}^n\frac{\SSS_{T,i,j}^2}{\SSS_{T,j,j}S_{T,i,i}S_{T,j,j}}\cdot\left|\frac2T\VV_i^\prime W_i-\frac1TW_i^\prime W_i\right|\\
	& =  O_p(n^{\frac{1}{2r_1}+\epsilon-1/2})\cdot \frac1T \sum\limits_{i\not=j}^n\SSS_{T,i,j}^2.
	\end{aligned}$$
	By Lemma~\ref{lem:trace2},  $\frac1T\sum\limits_{i\not=j}^n\SSS_{T,i,j}^2=O_p(1)$. 
	Therefore, the first term has the order of $O_p(n^{\frac{1}{2r_1}+\epsilon-1/2})$.
	
	The second term is estimated similarly with a same order, so that
	we conclude that $\sum\limits_{i\not=j}^n A_{ij}=O_p(n^{\frac{1}{2r_1}+\epsilon_1-1/2})\cvp0$.

	Next, $\sum\limits_{i\not=j}^nB_{ij}=o_p(1)$ holds with Lemma \ref{lem:trace2}.	
\end{proof}

\begin{proposition}\label{PET_prop}
	Under Assumption (B1)-(B2),  we have
	\[ tr\left(\hat{R}_T^4-R_T^4\right) =o_p(1).
	\]
\end{proposition}

\begin{proof}
	It is easy to verify that
	$$\begin{aligned}
	tr\left(\hat{R}_T^4-R_T^4\right)&=\sum\limits_{i,j,l}^n\left( \frac{\SSS_{T,i,j}^2\SSS_{T,j,l}^2}{\SSS_{T,i,i}\SSS_{T, j,j}^2\SSS_{T,l,l}}-\frac{S_{T,i,j}^2 S_{T,j,l}^2}{S_{T,i,i}S_{T, j,j}^2S_{T,l,l}}\right)\\
	&+2\sum\limits_{i,j,l}^n\sum\limits_{s>l}^n \left(\frac{\SSS_{T,i,j}\SSS_{T,j,l}\SSS_{T,i,s}\SSS_{T,s,l}}{\SSS_{T,i,i}\SSS_{T, j,j}\SSS_{T,l,l}\SSS_{T,s,s}}- \frac{S_{T,i,j}S_{T,j,l}S_{T,i,s}S_{T,s,l}}{S_{T,i,i}S_{T, j,j}S_{T,l,l}S_{T,s,s}} \right)\\
	&=\sum\limits_{i,j,l}^n \left(A_{ijl}+A^\ast_{ijl}\right)\nonumber +2\sum\limits_{i,j,l}^n\sum\limits_{s>l}^n \left(B_{ijls}+B^\ast_{ijls}\right)
	\end{aligned}$$  
	where
	
	\[A_{ijl} = \frac{\SSS_{T,i,j}^2\SSS_{T,j,l}^2}{\SSS_{T,i,i}\SSS_{T, j,j}^2\SSS_{T,l,l}}-\frac{\SSS_{T,i,j}^2 \SSS_{T,j,l}^2}{S_{T,i,i}S_{T, j,j}^2S_{T,l,l}},\quad A^\ast_{ijl} = \frac{\SSS_{T,i,j}^2 \SSS_{T,j,l}^2-S_{T,i,j}^2 S_{T,j,l}^2}{S_{T,i,i}S_{T, j,j}^2S_{T,l,l}}, \]
	\[B_{ijls} = \frac{\SSS_{T,i,j}\SSS_{T,j,l}\SSS_{T,i,s}\SSS_{T,s,l}}{\SSS_{T,i,i}\SSS_{T, j,j}\SSS_{T,l,l}\SSS_{T,s,s}}- \frac{\SSS_{T,i,j}\SSS_{T,j,l}\SSS_{T,i,s}\SSS_{T,s,l}}{S_{T,i,i}S_{T, j,j}S_{T,l,l}S_{T,s,s}}, \] 
	\[B^\ast_{ijls}= \frac{\SSS_{T,i,j}\SSS_{T,j,l}\SSS_{T,i,s}\SSS_{T,s,l}-S_{T,i,j}S_{T,j,l}S_{T,i,s}S_{T,s,l}}{S_{T,i,i}S_{T, j,j}S_{T,l,l}S_{T,s,s}}. \]

	We calculate $\sum\limits_{i\neq j}^n A_{ij}$ first. Note that
	\[A_{ijl} = \SSS_{T,i,j}^2\SSS_{T,j,l}^2 \left( \frac{S_{T,i,i}-\SSS_{T,i,i}}{\SSS_{T,i,i}\SSS_{T, j,j}^2\SSS_{T,l,l}S_{T,i,i}} + \frac{S_{T,j,j}^2-\SSS_{T, j,j}^2}{\SSS_{T, j,j}^2\SSS_{T,l,l}S_{T,i,i}S_{T,j,j}^2} +\frac{S_{T,l,l}-\SSS_{T,l,l}}{\SSS_{T,l,l}S_{T,i,i}S_{T,j,j}^2S_{T,l,l}} \right). \]
	For the first term, using Lemma~\ref{lem:estimates}
	$$\begin{aligned}
	\left| \sum_{i, j, l}^n \SSS_{T,i,j}^2\SSS_{T,j,l}^2 \frac{S_{T,i,i}-\SSS_{T,i,i}}{\SSS_{T,i,i}\SSS_{T, j,j}^2\SSS_{T,l,l}S_{T,i,i}} \right|&= \left|\sum_{i, j, l}^n \frac{\SSS_{T,i,j}^2\SSS_{T,j,l}^2}{\SSS_{T,i,i}\SSS_{T, j,j}^2\SSS_{T,l,l}S_{T,i,i}} \cdot \left(\frac2T\VV_i^\prime W_i-\frac1TW_i^\prime W_i\right)   \right| \\
	& \leq \sum_{i, j, l}^n \frac{\SSS_{T,i,j}^2\SSS_{T,j,l}^2}{\SSS_{T,i,i}\SSS_{T, j,j}^2\SSS_{T,l,l}S_{T,i,i}} \cdot \left|\frac2T\VV_i^\prime W_i-\frac1TW_i^\prime W_i \right| \\
	& = O_p(n^{\frac{1}{2r_1}+\epsilon-1/2}) \cdot \frac1T \sum_{i, j, l}^n\SSS_{T,i,j}^2\SSS_{T,j,l}^2.
	\end{aligned}$$
	By Corollary \ref{cor:diff}, $\frac1T \sum_{i, j, l}^n\SSS_{T,i,j}^2\SSS_{T,j,l}^2=O_p\left(1 \right)$. Therefore, the first term has the order of $O_p\left( n^{\frac{1}{2r_1}+\epsilon-1/2}\right)$.
	
	For the second term, 
	$$\begin{aligned}
	\left| \sum_{i, j, l}^n \SSS_{T,i,j}^2\SSS_{T,j,l}^2 \frac{S_{T,j,j}^2-\SSS_{T, j,j}^2}{\SSS_{T,i,i}\SSS_{T, j,j}^2\SSS_{T,l,l}S_{T,i,i}} \right|&= \left|\sum_{i, j, l}^n \frac{\SSS_{T,i,j}^2\SSS_{T,j,l}^2}{\SSS_{T,i,i}\SSS_{T, j,j}^2\SSS_{T,l,l}S_{T,i,i}} \cdot \left(S_{T,j,j}^2-\SSS_{T, j,j}^2\right)   \right|\\
	&\leq \sum_{i, j, l}^n \frac{\SSS_{T,i,j}^2\SSS_{T,j,l}^2}{\SSS_{T,i,i}\SSS_{T, j,j}^2\SSS_{T,l,l}S_{T,i,i}} \left|S_{T,j,j}^2-\SSS_{T, j,j}^2 \right|\\
	&= O_p\left(n^{\frac{1}{2r_1}+\epsilon-3/2} \right)\cdot \frac1T \sum_{i, j, l}^n\SSS_{T,i,j}^2\SSS_{T,j,l}^2\\
	&= O_p\left(n^{\frac{1}{2r_1}+\epsilon-3/2} \right).
	\end{aligned}$$

	The third term is estimated similarly with the first term with a same order, so that we conclude that $\sum_{i, j, l}^n A_{ijl}=O_p\left(  n^{\frac{1}{2r_1}+\epsilon-1/2}\right)+O_p\left(n^{\frac{1}{2r_1}+\epsilon-3/2} \right)\cvp0$.
	
	Next, $\sum_{i, j, l}^n A^\ast_{ijl}=o_p(1)$ holds with Corollary \ref{cor:diff}.
	
	Then we calculate $\sum\limits_{i,j,l}^n\sum\limits_{s>l}^n B_{ijls}$. Note that $$\begin{aligned}
	B_{ijls} =& \SSS_{T,i,j}\SSS_{T,j,l}\SSS_{T,i,s}\SSS_{T,l,s}\Bigg( \frac{S_{T,i,i}-\SSS_{T,i,i}}{\SSS_{T,i,i}\SSS_{T,j,j}\SSS_{T,l,l}\SSS_{T,s,s}S_{T,i,i}} +\frac{S_{T,j,j}-\SSS_{T,j,j}}{\SSS_{T,j,j}\SSS_{T,l,l}\SSS_{T,s,s}S_{T,i,i}S_{T,j,j}}\\
	& \hskip4cm \frac{S_{T,l,l}-\SSS_{T,l,l}}{\SSS_{T,l,l}\SSS_{T,s,s}S_{T,i,i}S_{T,j,j}S_{T,l,l}} +\frac{S_{T,s,s}-\SSS_{T,s,s}}{\SSS_{T,s,s}S_{T,i,i}S_{T,j,j}S_{T,l,l}S_{T,s,s}} \Bigg)
	\end{aligned}$$
	For the first term, we have
	$$\begin{aligned}
	\left|\sum\limits_{i,j,l}^n\sum\limits_{s>l}^n \SSS_{T,i,j}\SSS_{T,j,l}\SSS_{T,i,s}\SSS_{T,l,s}\frac{S_{T,i,i}-\SSS_{T,i,i}}{\SSS_{T,i,i}\SSS_{T,j,j}\SSS_{T,l,l}\SSS_{T,s,s}S_{T,i,i}} \right|\\\leq \sum\limits_{i,j,l}^n\sum\limits_{s>l}^n \frac{\left|\SSS_{T,i,j}\SSS_{T,j,l}\SSS_{T,i,s}\SSS_{T,l,s}\right|}{\SSS_{T,i,i}\SSS_{T,j,j}\SSS_{T,l,l}\SSS_{T,s,s}S_{T,i,i}}\cdot \left|\frac2T\VV_i^\prime W_i-\frac1TW_i^\prime W_i \right|\\
	= O_p\left(n^{\frac{1}{2r_1}+\epsilon-1/2} \right)\cdot \frac{1}{T} \sum\limits_{i,j,l}^n\sum\limits_{s>l}^n\left|\SSS_{T,i,j}\SSS_{T,j,l}\SSS_{T,i,s}\SSS_{T,l,s}\right|
	\end{aligned}$$
	By Corollary \ref{cor:diff}, $\frac{1}{T} \sum\limits_{i,j,l}^n\sum\limits_{s>l}^n\left|\SSS_{T,i,j}\SSS_{T,j,l}\SSS_{T,i,s}\SSS_{T,l,s}\right|=O_p(1)$, so the first term has the order of $ O_p\left(n^{\frac{1}{2r_1}+\epsilon-1/2} \right)$. The remaining terms are estimated similarly with same order, so that we conclude that $\sum\limits_{i,j,l}^n\sum\limits_{s>l}^n B_{ijls}\cvp 0$.
	
	Finally, $\sum\limits_{i,j,l}^n\sum\limits_{s>l}^n B^\ast_{ijls}=o_p(1)$ holds with Corollary \ref{cor:diff}. The proof is complete.
\end{proof}

\end{document}